\documentclass[a4paper,UKenglish,cleveref,autoref]{lipics-v2019-2cmMargin}

\title{Close relatives of Feedback Vertex Set without single-exponential algorithms parameterized by treewidth}

\titlerunning{Close relatives of FVS without single-exponential algorithms in treewidth}

\author{Benjamin Bergougnoux}{Department of Informatics, University of Bergen, Bergen, Norway}{benjamin.bergougnoux@uib.no}{https://orcid.org/0000-0002-6270-3663}{}

\author{\'{E}douard Bonnet}{Univ Lyon, CNRS, ENS de Lyon, Université Claude Bernard Lyon 1, LIP UMR5668, France}{edouard.bonnet@ens-lyon.fr}{https://orcid.org/0000-0002-1653-5822}{}

\author{Nick Brettell}{School of Mathematics and Statistics, Victoria University of Wellington, New Zealand}{nick.brettell@vuw.ac.nz}{https://orcid.org/0000-0002-1136-418X}{}

\author{O-joung Kwon}{Department of Mathematics, Incheon~National~University, Incheon,~South~Korea \\ Discrete Mathematics Group, Institute~for~Basic~Science~(IBS), Daejeon,~South~Korea}{ojoungkwon@gmail.com}{https://orcid.org/0000-0003-1820-1962}{}

\authorrunning{B. Bergougnoux, \'E. Bonnet, N. Brettell, and O. Kwon}

\Copyright{Benjamin Bergougnoux, Édouard Bonnet, Nick Brettell, and O-joung Kwon}



\category{}

\relatedversion{}

\supplement{}

\acknowledgements{This work was initiated while the authors attended the ``2019 IBS Summer research program on Algorithms and Complexity in Discrete Structures'', hosted by the IBS discrete mathematics group. %
The third author received support from the Leverhulme Trust (RPG-2016-258). %
The fourth author is supported by the National Research Foundation of Korea (NRF) grant funded by the Ministry of Education (No. NRF-2018R1D1A1B07050294) and supported by the Institute for Basic Science (IBS-R029-C1).}

\nolinenumbers 

\hideLIPIcs  

\EventEditors{John Q. Open and Joan R. Access}
\EventNoEds{2}
\EventLongTitle{}
\EventShortTitle{}
\EventAcronym{}
\EventYear{2020}
\EventDate{}
\EventLocation{}
\EventLogo{}
\SeriesVolume{42}
\ArticleNo{23}

\usepackage{amsmath,amssymb,amsthm}
\usepackage{xspace,xcolor}
\usepackage{tikz}
\usepackage{graphicx}
\usetikzlibrary{fit,automata,arrows}

\usepackage{booktabs,cite}
\usepackage{comment}

\newcommand\cC{\mathcal{C}}

\newcommand\cR{\mathcal{R}}
\newcommand\cP{\mathcal{P}}

\newcommand\cB{\mathcal{B}}

\newcommand\cA{\mathcal{A}}

\mathchardef\mhyphen="2D
\newcommand{\aux}{\text{Aux}}
\newcommand{\inc}{\text{Inc}}
\newcommand\abs[1]{\lvert #1\rvert}

\newtheorem{fact}[theorem]{Fact}

\newcommand{\shortfvs}{\textsc{FVS}\xspace}
\newcommand{\fvs}{\textsc{Feedback Vertex Set}\xspace}
\newcommand{\oct}{\textsc{Odd Cycle Transversal}\xspace}
\newcommand{\shortect}{\textsc{ECT}\xspace}
\newcommand{\ect}{\textsc{Even Cycle Transversal}\xspace}
\newcommand{\sfvs}{\textsc{Subset Feedback Vertex Set}\xspace}
\newcommand{\soct}{\textsc{Subset Odd Cycle Transversal}\xspace}

\newcommand{\ssfvs}{\textsc{Subset FVS}\xspace}
\newcommand{\ssoct}{\textsc{Subset OCT}\xspace}

\newcommand{\sfes}{\textsc{Restricted Edge-Subset Feedback Edge Set}\xspace}
\newcommand{\ssfes}{\textsc{RESFES}\xspace}
\newcommand{\mwc}{\textsc{Multiway Cut}\xspace}
\newcommand{\vmwc}{\textsc{Node Multiway Cut}\xspace}
\newcommand{\dmwc}{\textsc{Directed Multiway Cut}\xspace}
\newcommand{\tw}{\mathsf{tw}}
\newcommand{\pw}{\mathsf{pw}}

\newcommand{\Oh}{\ensuremath{\mathcal{O}}}

\newcommand{\NP}{\textsf{NP}\xspace}
\newcommand{\FPT}{\textsf{FPT}\xspace}

\newcommand{\reduce}{\mathsf{reduce}}

\newcommand{\defparproblem}[4]{
 \vspace{1mm}
\noindent\fbox{
 \begin{minipage}{0.96\textwidth}
 \begin{tabular*}{\textwidth}{@{\extracolsep{\fill}}lr} #1 & {\bf{Parameter:}} #3 \\ \end{tabular*}
 {\bf{Input:}} #2 \\
 {\bf{Question:}} #4
 \end{minipage}
 }
 \vspace{1mm}
}

\begin{document}

\maketitle

\begin{abstract}
  The Cut \& Count technique and the rank-based approach have lead to single-exponential \FPT algorithms parameterized by treewidth, that is, running in time $2^{\Oh(\tw)}n^{\Oh(1)}$, for \textsc{Feedback Vertex Set} and connected versions of the classical graph problems (such as \textsc{Vertex Cover} and \textsc{Dominating Set}). 
  We show that \sfvs, \soct, \sfes, \vmwc, and \mwc are unlikely to have such running times.
  More precisely, we match algorithms running in time $2^{\Oh(\tw \log \tw)}n^{\Oh(1)}$ with tight lower bounds under the Exponential-Time Hypothesis (ETH), ruling out $2^{o(\tw \log \tw)}n^{\Oh(1)}$, where $n$ is the number of vertices and $\tw$ is the treewidth of the input graph.
  Our algorithms extend to the weighted case, while our lower bounds also hold for the larger parameter pathwidth and do not require weights.
  We also show that, in contrast to \oct, there is no $2^{o(\tw \log \tw)}n^{\Oh(1)}$-time algorithm for \ect under the ETH.
\end{abstract}

\section{Introduction}\label{sec:intro}

Many \NP-hard graph problems admit polynomial-time algorithms on graphs with bounded \emph{treewidth}, a~measure of how well a graph accommodates a decomposition into a tree-like structure.
In fact, Courcelle's Theorem \cite{courcelle} states that any problem definable in MSO$_2$ logic can be solved in linear time on graphs of bounded treewidth.
To obtain a more fine-grained perspective on the dependence on treewidth for certain problems, it is useful to study the parameterized complexity with respect to treewidth.  In particular, we can ask: what is the ``smallest'' function $f$ for which we can obtain an algorithm that, given a graph with treewidth~$\tw$, has running time $f(\tw)n^{\Oh(1)}$?
For \fvs, standard dynamic programming techniques can be used to obtain an algorithm running in $2^{\Oh(\tw \log \tw)}n^{\Oh(1)}$ time, and for a while many believed this to be essentially best possible.
However, this changed in 2011 when Cygan et al.~\cite{Cygan11} developed the Cut\&Count technique, by which they obtained a \emph{single-exponential} $3^{\tw}n^{\Oh(1)}$-time randomized algorithm.
Following this, Bodlaender et al.~\cite{Bodlaender15} showed there is a deterministic $2^{\Oh(\tw)}n^{\Oh(1)}$-time algorithm, using a rank-based approach and the concept of representative sets.
The same year, Pilipczuk~\cite{Pilipczuk11} exhibited a logic fragment whose model checking admits a single-exponential algorithm parameterized by the treewidth of the input graph, thereby providing a scaled-down but more fine-grained version of Courcelle's theorem. 
Moreover, also in 2011, Lokshtanov et al.~\cite{Lokshtanov11} developed a framework yielding $2^{\Omega(\tw \log \tw)}n^{\Oh(1)}$-time lower bounds under the Exponential Time Hypothesis (ETH).
Recall that the ETH asserts that there is a real number $\delta > 0$ such that \textsc{$3$-SAT} cannot be solved in time $2^{\delta n}$ on $n$-variable formulas~\cite{ImpagliazzoP01}. 
Lokshtanov et al.'s paper prompted several authors to investigate the exact time-dependency on treewidth for a variety of graph modification problems.

For a \emph{vertex-deletion problem}, the task is to delete at most $k$ vertices so that the resulting graph is in some target class.
\fvs can be viewed as a vertex-deletion problem where the graphs in the target class consist of blocks with at most two vertices (a \emph{block} is a maximal subgraph $H$ such that $H$ has no cut vertices).
Bonnet et al.~\cite{BonnetBKM19} considered the class of problems, generalizing \fvs, where the target graphs are those consisting of blocks each of which has a bounded number of vertices, and is in some fixed hereditary, polynomial-time recognizable class $\mathcal{P}$.
They showed that such a problem is solvable in time $2^{\Oh(\tw)}n^{\Oh(1)}$ precisely when each graph in $\mathcal{P}$ is chordal (when $\mathcal{P}$ does not satisfy this condition, an algorithm with running time $2^{o(\tw \log \tw)}n^{\Oh(1)}$ would refute the ETH).
Baste et al.~\cite{Baste19} studied another generalization of \fvs: the vertex-deletion problem where the target graphs are those having no minor isomorphic to a fixed graph $H$.
They showed a single-exponential parameterized algorithm in treewidth is possible precisely when $H$ is a minor of the banner (the cycle on four vertices with a degree-1 vertex attached to it), but $H$ is not $P_5$ (the path graph on five vertices), assuming the ETH holds.

So-called \emph{slightly superexponential parameterized algorithms}, running in time $2^{\Oh(\tw \log \tw)}n^{\Oh(1)}$, are by no means a formality for problems that are \FPT in treewidth.
For instance, Pilipczuk~\cite{Pilipczuk11} showed that deciding if a graph has a transversal of size at most $k$ hitting all cycles of length exactly $\ell$ (or length at most $\ell$) for a fixed value $\ell$ cannot be solved in time $2^{o(\tw^2)} n^{\Oh(1)}$, unless the ETH fails.
This lower bound matches a dynamic-programming based algorithm running in time $2^{\Oh(\tw^2)} n^{\Oh(1)}$.
Cygan et al.~\cite{Cygan17} investigated the more general problem of hitting all subgraphs $H$ of a given graph $G$, for a fixed pattern graph $H$, again parameterized by treewidth.
For various $H$, they found algorithms running in time $2^{\Oh(\tw^{u(H)})} n^{\Oh(1)}$, and proved ETH lower bounds in $2^{\Omega(\tw^{\ell(H)})} n^{\Oh(1)}$, for values $1 < \ell(H) \leqslant u(H)$ depending on $H$.
Another recent example is provided by Sau and Uéverton~\cite{Sau20} who prove similar results for the analogous problem where ``subgraphs'' is replaced by ``induced subgraphs''.
Finally, for the vertex-deletion problem where the target class is a proper minor-closed class given by the non-empty list of forbidden minors, it is still open if the double-exponential dependence on treewidth is asymptotically best possible~\cite{Baste17}.

Sometimes, only a seemingly slight generalization of \fvs can result in problems with no single-exponential algorithm parameterized by treewidth.
Bonamy et al.~\cite{Bonamy18} showed that \textsc{Directed Feedback Vertex Set} can be solved in time $2^{\Oh(\tw \log \tw)}n^{\Oh(1)}$ but not faster under the ETH, where $\tw$ is the treewidth of the underlying undirected graph.
In this paper, we consider another collection of problems that generalize \fvs, and that do not have single-exponential algorithms parameterized by treewidth.
An equivalent formulation of \shortfvs is to find a transversal of \emph{all} cycles in a given graph.
We consider problems where the goal is to find a transversal of \emph{some subset} of the cycles of a given graph.
If this subset of cycles is those that intersect some fixed set of vertices~$S$, we obtain the following problem:

\defparproblem{\sfvs (\ssfvs)}{A graph $G$, a subset of vertices $S \subseteq V(G)$, and an integer $k$.}{$\tw(G)$}{Is there a set of at most $k$ vertices hitting all the cycles containing a vertex in $S$?}

If we further restrict this set of cycles to those that are odd, we obtain the next problem:

\defparproblem{\soct (\ssoct)}{A graph $G$, a subset of vertices $S \subseteq V(G)$, and an integer $k$.}{$\tw(G)$}{Is there a set of at most $k$ vertices hitting all the odd cycles containing a vertex in $S$?}

\noindent
Both of these problems are \NP-complete.
By setting $S = V(G)$, one sees that the latter problem generalizes \oct, for which Fiorini et al.~\cite{Fiorini08} presented a $2^{\Oh(\tw)}n^{\Oh(1)}$-time algorithm.

Alternatively, one can require a transversal of even cycles.  We first consider the problem of finding a transversal of \emph{all} even cycles since, to the best of our knowledge, the fine-grained complexity of this problem parameterized by treewidth has not previously been studied.

\defparproblem{\ect (\shortect)}{A graph $G$ and an integer $k$.}{$\tw(G)$}{Is there a set of at most $k$ vertices hitting all the even cycles of $G$?}

\noindent
We note that parameterizations by solution size have been studied for these three problems~\cite{CyganPPW13,kakimura2015fixed,lokshtanov2017hitting,MisraRRS12,Wahlstrom14}.

We now move to edge variants of \shortfvs.
Note that \textsc{Feedback Edge Set}, where the goal is to find a set of edges of weight at most $k$ that hits the cycles, can be solved in linear time, since it is equivalent to finding a maximum-weight spanning forest.
Xiao and Nagamochi showed that the subset variants \textsc{Vertex-Subset Feedback Edge Set} and \textsc{Edge-Subset Feedback Edge Set}, where the deletion set only needs to hit cycles containing a vertex or an edge (respectively) of a given set $S$, can also be solved in linear time~\cite{Xiao12}.
On the other hand, the latter problem becomes \NP-complete when the deletion set cannot intersect $S$.
This problem is known as \sfes.

\defparproblem{\sfes (\ssfes)}{A graph $G$, a subset of edges $S \subseteq E(G)$, and an integer $k$.}{$\tw(G)$}{Is there a set of at most $k$ edges of $E(G) \setminus S$ whose removal yields a graph without any cycle containing at least one edge of $S$?}

The final two \NP-complete problems we consider are closely related to \sfvs and \sfes, respectively (see the remark in \cref{parampres}).
They are well-established problems with an abundance of approximation and parameterized algorithms in the literature.

\defparproblem{\vmwc}{A graph $G$, a subset of vertices $T \subseteq V(G)$, called \emph{terminals}, and an integer $k$.}{$\tw(G)$}{Is there a set of at most $k$ vertices of $V(G) \setminus T$ hitting every path between a pair of terminals?}

\defparproblem{\mwc}{A graph $G$, a subset of vertices $T \subseteq V(G)$, called \emph{terminals}, and an integer $k$.}{$\tw(G)$}{Is there a set of at most $k$ edges hitting every path between a pair of terminals?}

The look-alike problem \textsc{Multicut}, where the task is to separate each pair of terminals in a given set of pairs (rather than all the pairs in a given set), is \NP-complete on trees~\cite{Garg97}.
Therefore a parameterization by treewidth cannot help here.
In the language of parameterized complexity, \textsc{Multicut} parameterized by treewidth is \textsf{paraNP}-complete.

\subsection{Our contribution}
\label{intro-contib}

With the exception of \ect, for which we provide only a lower bound, we show that all the problems formally defined so far admit a slightly superexponential parameterized algorithm, and that this running time cannot be improved, unless the ETH fails.
We leave as an open problem the existence of a slightly superexponential algorithm for \textsc{(Subset) Even Cycle Transversal} parameterized by treewidth. 
We note that Deng et al.~\cite{Deng13} have already shown that \mwc can be solved in time $2^{\Oh(\tw \log \tw)} n^{\Oh(1)}$.
Our algorithms work for treewidth and weights, while our lower bounds hold for the larger parameter pathwidth and do not require weights.

On the algorithmic side we show the following:
\begin{theorem}\label{alg:main}
 The following problems can be solved in time $2^{\Oh(\tw \log \tw)} n^{\Oh(1)}$ on $n$-vertex graphs with treewidth~$\tw$\emph{:}
  \begin{itemize}
  \item \sfvs,
  \item \soct,
  \item \sfes, and
  \item \vmwc.
  \end{itemize}
\end{theorem}

We provide algorithms having the claimed running time for the weighted versions of each of the four problems in Theorem~\ref{alg:main}.
In these weighted versions, the input graph is given with a weight function $w$ on the vertices when the problem is to find a set of vertices, or on the edges when the problem is to find a set of edges.
Furthermore, in the weighted versions, the problem asks for a solution of weight at most $k$.

On the complexity side, the main conceptual contribution of the paper is to show that problems seemingly quite close to \fvs do not admit a single-exponential algorithm parameterized by treewidth, under the ETH.
\begin{theorem}\label{hardness:main}
  Unless the ETH fails, the following problems cannot be solved in time $2^{o(\pw \log \pw)}n^{\Oh(1)}$ on $n$-vertex graphs with pathwidth $\pw$\emph{:}
  \begin{itemize}
  \item \sfvs,
  \item \soct,
  \item \ect,
  \item \sfes,
  \item \vmwc, and
  \item \mwc.
  \end{itemize}
\end{theorem}

For the last two problems, our reductions build instances where the number of terminals $|T|$ is $\Theta(\pw)$.
Thus we also rule out a running time of $|T|^{o(\pw)}$.
All the reductions are from \textsc{$k \times k$-(Permutation) Independent Set/Clique} following a strategy suggested by Lokshtanov et al.~\cite{Lokshtanov18} (see for instance,~\cite{Broersma13,Drange16,BonnetBKM19,Bonamy18,Baste19}).
These problems cannot be solved in time $2^{o(k \log k)}$, unless the ETH fails.

\defparproblem{\textsc{$k \times k$-Independent Set}}{A graph $H$ with vertex set $V(H)=[k]^2$ for some integer $k$.}{$k$}{An independent set of size $k$ hitting each column exactly once.}

\defparproblem{\textsc{$k \times k$-Permutation Independent Set}}{A graph $H$ with vertex set $V(H)=[k]^2$ for some integer $k$.}{$k$}{An independent set of size $k$ hitting each column and each row exactly once.}

A \emph{row} is a set of vertices of the form $\{(i,1),(i,2),\ldots, (i,k)\} \subset V(H)$ for some $i \in [k]$, while a column is a set $\{(1,j),(2,j),\ldots, (k,j)\} \subset V(H)$ for some $j \in [k]$.
The problem \textsc{$k \times k$-(Permutation) Clique} is defined analogously, where the solution is required to be a clique rather than an independent set.\footnote{Observe that we switch the columns and the rows compared to the original definition of \textsc{$k \times k$-Clique}~\cite{Lokshtanov18}.
While this is of course equivalent, it will make the representation of some gadgets slightly more conducive to the page layout.}

\subparagraph{Roadmap for the lower bounds.}
To prove~\cref{hardness:main}, we start by designing a gadget specification for generic vertex-deletion problems.
We show that any such problem, allowing for gadgets respecting the specification, has the lower bound given in~\cref{hardness:main}.
This is achieved by a meta-reduction from \textsc{$k \times k$-Permutation Independent Set}.
We give gadgets for \ssfvs, \ssoct, and \shortect that comply with the specification.
We thus obtain the first three items of the theorem in a unified way, with simple and reusable gadgets.
This mini-framework may in principle be useful for other vertex-deletion problems.

In order to show a stronger lower bound for \vmwc, with the number of terminals in $\Theta(k)$, we depart from the previous specification slightly, although we still use some shared notation and arguments to bound the pathwidth, where convenient.
This reduction is from \textsc{$k \times k$-Independent Set}.

Finally, the reduction to \mwc is more intricate.
For this problem it is surprisingly challenging to discourage the undesirable solutions ``cutting close'' to every terminal but one, where the deletion set yields a very large connected component for one terminal, and small components for the rest of the terminals.
In particular, the trick used for the \vmwc lower bound cannot be replicated.
We overcome this issue by designing a somewhat counter-intuitive edge gadget which encourages the retention of as many pairs of endpoints linked to two (distinct) terminals as possible.
This uses the simple fact that, in a $\Delta$-regular graph, a clique of size $k$ minimizes the number of edges covered by $k$ vertices: $\Delta k - {k \choose 2}$ vs $\Delta k$ for an independent set of size $k$. 
We then reduce from \textsc{$k \times k$-Permutation Clique}.
We discuss why getting the same lower bound for a regular variant of \textsc{$k \times k$-Permutation Clique} is technical, and bypass that difficulty by encoding a \emph{degree-equalizer} gadget directly in the \mwc instance.
As a side note, we nevertheless prove that a semi-regular variant of \textsc{$k \times k$-Clique} also has the slightly superexponential lower bound.
This proof uses a constructive version of the Hajnal-Szemerédi theorem on equitable colorings.


\subparagraph{A remark on parameter-preserving reductions between the problems.}\label{parampres}
  There is an easy reduction from \vmwc to \textsc{Weighted Subset Feedback Vertex Set} (\textsc{WSFVS}, for short).
  It consists of adding a vertex $v$ of ``infinite'' weight adjacent to all the terminals of the \mwc instance, which also all get ``infinite'' weight.
  The set $S$ of the \textsc{WSFVS} instance is $\{v\}$.
  The same process yields a reduction from \mwc to \sfes, where now the set $S$ of the \sfes instance contains all the edges incident to $v$ (recall that these edges are thus undeletable).

  From the latter reduction, we can immediately derive the lower bound for \sfes from the lower bound for \mwc (see~\cref{hardness:sfes}).
  However, the hardness result for \vmwc (see \cref{hardness:vmwc}) does not imply anything for \sfvs.
  Indeed, to encode the ``infinite'' weight that makes $v$ and its neighbors undeletable, one would have to duplicate these vertices many times.
  This would result in a large biclique, or at the very least a large biclique minor, and would thereby make the pathwidth or treewidth large.
  Therefore \cref{hardness:adhoc} is necessary and cannot be obtained by a simple modification of \cref{hardness:vmwc}.
  Finally, we observe that the straightforward reduction from \mwc to \vmwc requires vertex weights, or blows up the treewidth.
  So again \cref{hardness:vmwc} cannot be derived from \cref{hardness:mwc}.
  
  \subparagraph{Roadmap for the algorithms.}
  To prove Theorem~\ref{alg:main}, we first present a $2^{\Oh(\tw \log \tw)}n^3$-time algorithm for the weighted variant of \ssoct. 
  With a few modifications, this algorithm can solve the weighted variant of \ssfvs.
  We obtain algorithms for the other problems in Theorem~\ref{alg:main} by reducing these problems to the weighted variant of \ssfvs.

Let us explain our approach for \ssoct on a graph $G$ with $S\subseteq V(G)$.
We solve \ssoct indirectly by finding a set $X\subseteq V(G)$ of maximum weight that induces a graph with no odd cycles traversing $S$ (we call such a graph $S$-bipartite).
We prove that a graph has no odd cycle traversing $S$ if and only if for each block $C$, either $C$ is bipartite or $C$ has no vertex in $S$.
From this characterization, we prove that it is enough to store $2^{\Oh(\tw\log\tw)}$ partial solutions at each bag $B$ of a tree decomposition.

Let $B$ be a bag of the tree decomposition of $G$ and $G_B$ be the graph induced by the vertices in $B$ and its descendant bags in the tree decomposition.
A partial solution of $G_B$ is a set $X\subseteq V(G_B)$ that induces an $S$-bipartite graph.
We design an equivalence relation $\equiv_B$ on the partial solutions of $G_B$ such that for every $X\equiv_B Y$ and $W\subseteq V(G)\setminus V(G_B)$, $G[X\cup W]$ is $S$-bipartite if and only if $G[Y\cup W]$ is $S$-bipartite.
Consequently, it is enough to keep a partial solution of maximum weight for each equivalence class of $\equiv_B$.
Intuitively, the equivalence relation $\equiv_B$ is based on the information: (1) how the blocks of $G[X]$ intersecting $B$ are connected, (2) whether important blocks (that have the possibility to create an $S$-traversing odd cycle later) contain a vertex of $S$, and (3) the parity of the paths between the vertices in $B$.
Since $\equiv_B$ has $2^{\Oh(\tw \log \tw)}$ equivalence classes, we deduce from this equivalence relation a $2^{\Oh(\tw \log \tw)}n^3$-time algorithm with standard dynamic programming operations.
The polynomial factor $n^3$ appears because we can test $X\equiv_B Y$ in time $\Oh(n^2)$.

For the weighted variant of \ssfvs, we can use the same equivalence relation without (3).
We reduce the weighted variant of \vmwc to \ssfvs as explained in the previous subsection: 
by adding a vertex $v$ of infinite weight adjacent to the set of terminals, setting $S=\{v\}$, and also giving infinite weights to the terminals.
Furthermore, we reduce the weighted variant of \sfes to the weighted variant of \ssfvs
by subdividing each edge, setting $S$ as the set of subdivided vertices corresponding to the given subset of edges, and giving infinite weights to the original vertices and the vertices in $S$.
These two reductions show that both problems admit $2^{\Oh(\tw \log \tw)}n^3$-time algorithms.

\subsection{Organization of the paper}
The rest of the paper is organized as follows.
In~\cref{sec:preliminaries} we give the required graph-theoretic definitions and notation.
In~\cref{sec:lower-bounds} we prove all the ETH lower bounds of~\cref{hardness:main}.
More precisely, in~\cref{subsec:generic} we introduce a gadget specification for a generic vertex-deletion problem, and we show the slightly superexponential lower bound for any problem complying with the gadget specification.
In~\cref{subsec:gadgets} we design gadgets for \ssfvs, \ssoct, \shortect, and thus obtain the first three items of~\cref{hardness:main}.
In~\cref{subsec:vmwc-lower,subsec:mwc-lower} we present specific reductions for \vmwc and \mwc, respectively.
In Section~\ref{sec:algorithms} we prove that 
the weighted variants of \ssoct, \ssfvs, \sfes, and \vmwc admit $2^{\Oh(\tw \log \tw)}n^3$-time algorithms.

\section{Preliminaries}\label{sec:preliminaries}

We assume all graphs have no loops or parallel edges.
Let $G$ be a graph.
We denote the vertex set and the edge set of $G$ by $V(G)$ and $E(G)$, respectively.
For a vertex $v$ in $G$, we use $G-v$ to denote the \emph{deletion} of $v$ from $G$, that is, the graph obtained by removing $v$ and its incident edges.
For $X\subseteq V(G)$, we denote by $G-X$ the graph obtained by removing all vertices in $X$ and their incident edges.
For $X\subseteq V(G)$, we denote by $G[X]$ the subgraph induced by the vertex set $X$.
A subgraph $H$ of $G$ is an \emph{induced subgraph} of $G$ if $H=G[X]$ for some vertex subset $X$ of $G$.
For two graphs $G_1$ and $G_2$,  $G_1\cup G_2$ is the graph with the vertex set $V(G_1) \cup V(G_2)$ and the edge set $E(G_1) \cup E(G_2)$, 
and $G_1\cap G_2$ is the graph with the vertex set $V(G_1) \cap V(G_2)$ and the edge set $E(G_1) \cap E(G_2)$.
A set $X \subseteq V(G)$ is a \emph{clique} if $G$ has an edge between every pair of vertices in $X$.  A graph with vertex set $X \cup Y$ that has an edge between every vertex $x \in X$ and $y \in Y$ is called a \emph{biclique}, and is denoted $K_{|X|,|Y|}$.

For a vertex $v$ in $G$, we denote by $N_G(v)$ the set of neighbors of $v$ in $G$, 
and $N_G[v]:=N_G(v)\cup \{v\}$.
For $X\subseteq V(G)$, we let $N_G(X):=(\bigcup_{v \in X} N_G(v)) \setminus X$, and say $N_G(X)$ is the \emph{(open) neighborhood} of $X$.
For $u,v \in V(G)$, we say that $u$ and $v$ are \emph{twins} if $N(u) = N(v)$.
If $N[u] = N[v]$, then we also say that $u$ and $v$ are \emph{true} twins; whereas when $u$ and $v$ are non-adjacent twins, we say that $u$ and $v$ are \emph{false} twins.

A vertex $v$ of $G$ is a {\em cut vertex} if the deletion of $v$ from $G$ increases the number of connected components. 
We say $G$ is \emph{2-connected} if it is connected and has no cut vertices.
Note that every connected graph on at most two vertices is 2-connected.
A \emph{block} of $G$ is a maximal 2-connected subgraph of $G$.

Let $G$ be a graph.
A \emph{walk} in $G$ is a sequence of vertices where every consecutive pair of vertices is an edge of $G$.
The first and last vertices in a walk are called \emph{end-vertices}.
A walk is \emph{closed} if its two end-vertices are the same. 
Given two walks $W_1=(v_1,\dots,v_t)$ and $W_2=(v_t,v_{t+1},\dots,v_k)$ whose internal vertices are pairwise distinct, we denote by $W_1\cdot W_2$ the walk $(v_1,\dots,v_t,v_{t+1},\dots,v_k)$.
We say that a walk is odd (resp. even) if the number of edges used by the the walk is odd (resp. even).
Given $S\subseteq V(G)$, we say that a walk is \emph{$S$-traversing} if it contains at least one vertex in $S$. 
For a graph $H$ and subgraph $B$ of $H$, we say that a walk $W$ in $H$ is a $B$-walk if the endpoints of $W$ are in $B$ and the internal vertices of $W$ are not in $B$.
A \emph{path} of a graph is a walk where each vertex is used at most once.
A \emph{cycle} of a graph is a closed walk where each vertex, except the end-vertices, is used at most once.

\subsection{Treewidth}

A \emph{tree decomposition} of a graph $G$ is 
a pair $(T,\cB)$ 
consisting of a tree $T$
and a family $\cB=\{B_t\}_{t\in V(T)}$ of sets $B_t\subseteq V(G)$,
called \emph{bags},
satisfying the following three conditions:
\begin{enumerate}
	\item $V(G)=\bigcup_{t\in V(T)}B_t$,
	\item for every edge $uv$ of $G$, there exists a node $t$ of $T$ such that $u,v\in B_t$, and
	\item for $t_1,t_2,t_3\in V(T)$, $B_{t_1}\cap B_{t_3}\subseteq B_{t_2}$ whenever $t_2$ is on the path from $t_1$ to $t_3$ in $T$.
\end{enumerate}
The \emph{width} of a tree decomposition $(T,\cB)$ is $\max\{ \abs{B_{t}}-1:t\in V(T)\}$.	
The \emph{treewidth} of $G$ is the minimum width over all tree decompositions of $G$. 
A \emph{path decomposition} is a tree decomposition $(P,\mathcal{B})$ where $P$ is a path.
The \emph{pathwidth} of $G$ is the minimum width over all path decompositions of $G$. 
We denote a path decomposition $(P,\mathcal{B})$ as $(B_{v_1},\dotsc,B_{v_t})$, where $P$ is a path $v_1v_2\dotsb v_t$.

To design a dynamic programming algorithm, we use a convenient form of a tree decomposition known as a nice tree decomposition.
A tree $T$ is said to be \emph{rooted} if it has a specified node called the \emph{root}.
Let $T$ be  a rooted tree with root node $r$.
A node $t$ of $T$ is called a \emph{leaf} node if it has degree one and it is not the root.
For two nodes $t_1$ and $t_2$ of $T$, $t_1$ is a \emph{descendant} of $t_2$ if the unique path from $t_1$ to $r$ contains $t_2$. 
If a node $t_1$ is a descendant of a node $t_2$ and $t_1t_2\in E(T)$, 
then $t_1$ is called a \emph{child} of $t_2$.

A tree decomposition $(T,\cB=\{B_t\}_{t\in V(T)})$ is a \emph{nice tree decomposition} with root node $r\in V(T)$ if $T$ is a rooted tree with root node $r$, and every node $t$ of $T$ is one of the following:
\begin{enumerate}
	\item a \emph{leaf node}: $t$ is a leaf of $T$ and $B_t=\emptyset$;
	\item an \emph{introduce node}: $t$ has exactly one child $t'$ and $B_t=B_{t'}\cup \{v\}$ for some $v\in V(G)\setminus B_{t'}$;
	\item a \emph{forget node}: $t$ has exactly one child $t'$ and $B_t=B_{t'}\setminus \{v\}$ for some $v\in B_{t'}$; or
	\item a \emph{join node}: $t$ has exactly two children $t_1$ and $t_2$, and $B_t=B_{t_1}=B_{t_2}$.
\end{enumerate}

\begin{theorem}[Bodlaender et al.\ \cite{BodlaenderDDFLP16}]\label{thm:approxtw}
	Given an $n$-vertex graph $G$ and a positive integer $k$, one can  either output
	a tree decomposition of $G$ with width at most $5k+4$, or correctly answer that 
	the treewidth of $G$ is larger than $k$, in time $2^{\mathcal{O}(k)} n$.
\end{theorem}

\begin{lemma}[folklore; see Lemma 7.4 in \cite{CyganFKLMPPS15}]\label{lem:nicetd}
	Given a tree decomposition of an $n$-vertex graph $G$ of width $w$, 
	one can construct a nice tree decomposition $(T, \cB)$ of width $w$ with $\abs{V(T)}=\mathcal{O}(wn)$  in time $\mathcal{O}(k^2\cdot \max (\abs{V(T)}, \abs{V(G)}))$.
\end{lemma}

\subsection{Boundaried graphs}

For a graph $G$ and $X\subseteq V(G)$, the pair $(G, X)$ is called a \emph{boundaried graph}.
Two boundaried graphs $(G,X)$ and $(H,X)$ are said to be \emph{compatible} 
if $V(G-X)\cap V(H-X)=\emptyset$ and $G[X]=H[X]$.
For two compatible boundaried graphs $(G,X)$ and $(H,X)$, 
the \emph{sum} of two graphs is the graph obtained from the disjoint union of $G$ and $H$ 
by identifying each vertex of $X$ in $G$ with the same vertex in $H$ and 
removing an edge from multiple edges that appear in $X$.
We denote the resulting graph by $(G,X)\oplus (H,X)$.
See Figure~\ref{fig:sumgraphs} for an example.

\begin{figure}
	\centering
	\begin{tikzpicture}[scale=0.7]
	\tikzstyle{w}=[circle,draw,fill=black,inner sep=0pt,minimum width=4pt]

	\draw (1, 1) node [w] (v1) {};
	\draw (2, 1) node [w] (v2) {};
	\draw (3, 1) node [w] (v3) {};
	\draw (1, 2.5) node [w] (v4) {};
	\draw (2, 2) node [w] (v5) {};
	\draw (2, 3) node [w] (v6) {};
	
	\draw(v1)--(v2)--(v3);
	\draw(v1)--(v4)--(v5)--(v1);
	\draw(v4)--(v6)--(v5)--(v2);
	\draw(v5)--(v3);
	
	\draw (1, -1) node [w] (w1) {};
	\draw (2, -1) node [w] (w2) {};
	\draw (3, -1) node [w] (w3) {};
	\draw (1, -2) node [w] (w4) {};
	\draw (2, -2) node [w] (w5) {};
	\draw (3, -2) node [w] (w6) {};
	\draw (2.5, -2.6) node [w] (w7) {};
	
	\draw(w1)--(w2)--(w3);
	\draw(w4)--(w1)--(w5)--(w2)--(w6)--(w3);
	\draw(w4)--(w5)--(w6)--(w7)--(w5);
	
	\draw[rounded corners] (0.7, 1)--(0.7,1.3)--(3.3,1.3)--(3.3,0.7)--(0.7, 0.7)--(0.7,1);
	\draw[rounded corners] (0.7, -1)--(0.7,-1.3)--(3.3,-1.3)--(3.3,-0.7)--(0.7, -0.7)--(0.7,-1);

	\node at (-1, 2) {$(G,X)$};
	\node at (-1, -2) {$(H,X)$};

	\draw (8+1, 1-1) node [w] (x1) {};
	\draw (8+2, 1-1) node [w] (x2) {};
	\draw (8+3, 1-1) node [w] (x3) {};
	\draw (8+1, 2.5-1) node [w] (x4) {};
	\draw (8+2, 2-1) node [w] (x5) {};
	\draw (8+2, 3-1) node [w] (x6) {};
	
	\draw(x1)--(x2)--(x3);
	\draw(x1)--(x4)--(x5)--(x1);
	\draw(x4)--(x6)--(x5)--(x2);
	\draw(x5)--(x3);
	
	\draw (8+1, -1+1) node [w] (y1) {};
	\draw (8+2, -1+1) node [w] (y2) {};
	\draw (8+3, -1+1) node [w] (y3) {};
	\draw (8+1, -2+1) node [w] (y4) {};
	\draw (8+2, -2+1) node [w] (y5) {};
	\draw (8+3, -2+1) node [w] (y6) {};
	\draw (8+2.5, -2.6+1) node [w] (y7) {};
	
	\draw(y1)--(y2)--(y3);
	\draw(y4)--(y1)--(y5)--(y2)--(y6)--(y3);
	\draw(y4)--(y5)--(y6)--(y7)--(y5);
	
	\node at (11-1, -2.6) {$(G,X)\oplus (H,X)$};
	
	\draw[rounded corners] (8+0.7, 1-1)--(8+0.7,1.3-1)--(8+3.3,1.3-1)--(8+3.3,0.7-1)--(8+0.7, 0.7-1)--(8+0.7,1-1);

	\end{tikzpicture}     \caption{An example of the sum $(G,X)\oplus (H,X)$.}\label{fig:sumgraphs}
\end{figure}
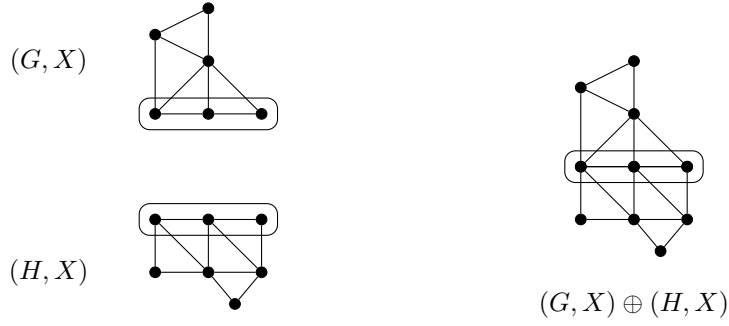

\section{Superexponential lower bounds parameterized by treewidth}\label{sec:lower-bounds}

Our reductions for \ssfvs, \ssoct, and \shortect, in~\cref{subsec:gadgets}, will have the same skeleton.
In order to avoid repeating the same arguments, we show in \cref{subsec:generic} the lower bound of \cref{hardness:main} for a meta-problem.
We prove the lower bound for \vmwc in \cref{subsec:vmwc-lower}, and the lower bounds for \mwc and \sfes in \cref{subsec:mwc-lower}.

\subsection{Lower bound for a generic vertex-deletion problem}\label{subsec:generic}

The scope of application of \cref{hardness:main} is any \emph{hereditary} vertex-deletion problem $\Pi$; that is, if $G-X$ satisfies a problem instance $P(\Pi)$, then $G-X'$ also satifies $P(\Pi)$ for every $X' \supseteq X$.
The main part of the input is a graph $G$ and a non-negative integer $k'$.
In addition, we allow any sort of labelings of $G$, be it subsets of vertices $S_1, S_2, \ldots \subseteq V(G)$, of edges $E_1, E_2, \ldots, \subseteq E(G)$, pairs of vertices $P_1, P_2, \ldots \subseteq {V(G) \choose 2}$, etc. The goal is to find a subset $X \subseteq V(G)$ of $k'$ vertices such that a property $P(\Pi)$, dependent on $\Pi$, is satisfied on $G - X$ with its induced labeling.
A subset of vertices $A \subseteq V(G)$ is a \emph{$\Pi$-obstruction} if $G[A]$ does not satisfy $P(\Pi)$.
A set $X \subseteq V(G)$ is \emph{$\Pi$-legal} if $G-X$ satisfies $P(\Pi)$ (in particular, solutions are $\Pi$-legal sets of size $k'$).
As $P(\Pi)$ is assumed hereditary, a $\Pi$-legal set intersects every \emph{$\Pi$-obstruction}.
Finally a \emph{$\Pi$-legal $s$-deletion within $Y$} is a set $X \subseteq Y$ of size at most $s$ such that $G[Y \setminus X]$ satisfies $P(\Pi)$.

\paragraph*{Common base}
The meta-result of \cref{hardness:generic} concerns hereditary vertex-deletion problems admitting four types of gadgets.
These gadgets, which will eventually depend on $\Pi$, are attached to a common problem-independent base.
We first describe the common base.
$H_{\bullet}$ is a set of $2k^2$ vertices, for some implicit positive integer $k$.
We denote these vertices by $v_{\bullet}(i,j,z)$ for each $i \in [k]$, $j \in [k]$, and $z \in [2]$.
We imagine the vertices of $H_{\bullet}$ being displayed in a $k$-by-$k$ grid with $v_{\bullet}(i,j,1)$ and $v_{\bullet}(i,j,2)$ side by side in the $i$-th row and $j$-th column.

The \emph{base} consists of copies of $H_{\bullet}$ that we denote by $H_1, H_2, \ldots$ and typically index by $p$.
The vertices of $H_p$ are denoted by $v_p(i,j,z)$.
The vertices $v_p(i,j,1)$ and $v_p(i,j,2)$ are said to be \emph{homologous}.
We set $C_{p,j} := \bigcup_{i \in [k], z \in [2]} \{v_p(i,j,z)\}$ and refer to it as the \emph{$j$-th column} of $H_p$.
Similarly $R_{p,i} := \bigcup_{j \in [k], z \in [2]} \{v_p(i,j,z)\}$ is called the \emph{$i$-th row} of $H_p$.
We can attach to the base a list of gadgets as detailed now.
The vertices added to the base are called \emph{additional} or \emph{new}.

\paragraph*{Column selector gadget}
A~\emph{$k$-column selector} gadget has the following specification.
Its vertex set is a single column $C_{p,j}$ plus $\Oh(k)$ additional vertices $\mathcal C_{\text{sel}}(p,j)$.
The only restriction on the edge set of the gadget is that homologous vertices should remain non-adjacent.
Other than that, any edge can be added within $C_{p,j}$.
However the open neighborhood of $\mathcal C_{\text{sel}}(p,j)$ has to be contained in $C_{p,j}$.

A problem $\Pi$ \emph{admits a column selector gadget} if, for every positive integer $k$, one can build in time $k^{\Oh(1)}$ a $k$-column selector such that the only $\Pi$-legal $(2k-2)$-deletions within $C_{p,j} \cup \mathcal C_{\text{sel}}(p,j)$ are one of the $k$ sets: $C_{p,j} \setminus \{v_p(1,j,1),v_p(1,j,2)\}, C_{p,j} \setminus \{v_p(2,j,1),v_p(2,j,2)\}, \ldots, C_{p,j} \setminus \{v_p(k,j,1),v_p(k,j,2)\}$.

\paragraph*{Row selector gadget}
In order to keep small balanced separators, our~\emph{$k$-row selector} gadget is quite different from the $k$-column selector.
Its vertex set is a single row $R_{p,i}$ plus $\Oh(1)$ additional vertices $\mathcal R_{\text{sel}}(p,i)$.
Furthermore \emph{no} edge can be added within $R_{p,i}$.
Again the open neighborhood of $\mathcal R_{\text{sel}}(p,i)$ has to be contained in $R_{p,i}$.

A problem $\Pi$ \emph{admits a row selector gadget} if, for every positive integer $k$, one can build in time $k^{\Oh(1)}$ a $k$-row selector such that, for every $j \neq j' \in [k]$, $\mathcal R_{\text{sel}}(p,i) \cup \{v_p(i,j,1),v_p(i,j,2),v_p(i,j',1),v_p(i,j',2)\}$ is a $\Pi$-obstruction.

\paragraph*{Edge gadget}
The vertex set of an \emph{edge gadget} is of the form $\{v_p(i,j,1),v_p(i,j,2),v_p(i',j',1),v_p(i',j',2)\} \cup \mathcal E_p(i,j,i',j')$ where $i \neq i' \in [k]$, $j \neq j' \in [k]$, and $\mathcal E_p(i,j,i',j')$ is a set of $\Oh(k)$ vertices\footnote{$\Oh(1)$ vertices will actually suffice for all the gadgets of~\cref{subsec:gadgets}.}.
There is no restriction on the edge set.
As usual the open neighborhood of $\mathcal E_p(i,j,i',j')$ has to be contained in $\{v_p(i,j,1),v_p(i,j,2),v_p(i',j',1),v_p(i',j',2)\}$.

A problem $\Pi$ \emph{admits an edge gadget} if one can build in time $k^{\Oh(1)}$ an edge gadget such that $\mathcal E_p(i,j,i',j') \cup \{v_p(i,j,1),v_p(i,j,2),v_p(i',j',1),v_p(i',j',2)\}$ is a $\Pi$-obstruction.

\paragraph*{Propagation gadget}
The vertex set of a \emph{propagation gadget} is of the form $H_p \cup H_{p+1} \cup \mathcal P_p$ where $\mathcal P_p$ is a set of $k^{\Oh(1)}$ vertices.
There is a subset $\mathcal P'_p \subseteq \mathcal P_p$ of size $\Oh(k)$ such that each vertex of $\mathcal P_p \setminus \mathcal P'_p$ has at most one neighbor in $H_p \cup H_{p+1}$ and the rest of its neighborhood in $\mathcal P'_p$.
This fairly technical condition aims to give some extra flexibility while keeping sufficiently small separators between $H_p$ and $H_{p+1}$.
In particular, if $\mathcal P_p$ is itself of size $\Oh(k)$, then the condition is trivially met with $\mathcal P'_p = \mathcal P_p$ .
The propagation gadget has no edge with both endpoints in $H_p \cup H_{p+1}$.
Everything else is permitted, but the open neighborhood of $\mathcal P_p$ has to be contained in $H_p \cup H_{p+1}$.

A problem $\Pi$ \emph{admits a propagation gadget} if one can build in time $k^{\Oh(1)}$ a propagation gadget such that for every $i,j \neq j' \in [k]$, $\mathcal P_p \cup \{v_p(i,j,1),v_p(i,j,2),v_{p+1}(i,j',1),v_{p+1}(i,j',2)\}$ is a $\Pi$-obstruction.

\paragraph*{Intended-solution property}

A hereditary vertex-deletion problem $\Pi$ and a description of the four above gadgets for $\Pi$ have the \emph{intended-solution property} if the following holds.
On any graph $G$ built by adding to the base $H_1 \cup \ldots \cup H_p \cup \ldots H_m$ at most one edge gadget in each $H_p$, one propagation gadget between \emph{consecutive} pairs $H_p$ and $H_{p+1}$, and some column and row selector gadgets, every deletion set $\bigcup_{p \in [m], i \in [k],j \in [k] \setminus \{j_i\}, z \in [2]} \{v_p(i,j,z)\}$ (with $\{j_1, j_2, \ldots, j_k\} = [k]$) intersecting every edge gadget is $\Pi$-legal. 

We can now state the lower bound for the generic hereditary vertex-deletion problems.
\begin{theorem}\label{hardness:generic}
  Unless the ETH fails, every vertex-deletion problem $\Pi$ admitting a column selector, a row selector, an edge, and a propagation gadget, satisfying the intended-solution property, cannot be solved in time $2^{o(\pw \log \pw)}n^{\Oh(1)}$ on $n$-vertex graphs with pathwidth $\pw$.
\end{theorem}
\begin{proof}
  From any instance $H$ of \textsc{$k \times k$-Permutation Independent Set}, we build an equivalent $\Pi$-instance $(G,k'=k^{\Oh(1)})$ of size $k^{\Oh(1)}$ with pathwidth in $\Oh(k)$.
  Since under the ETH there is no algorithm solving \textsc{$k \times k$-Permutation Independent Set} in time $2^{o(k \log k)}k^{\Oh(1)}$, we derive the claimed lower bound.
  
  \subparagraph*{Construction.}
  We number the edges in $E(H)$ as $e_1, \ldots, e_m$. 
  We start with a base consisting of $m$ copies of $H_\bullet$, labelled $H_p$ for $p \in [m]$ (see description of the common base).
  The vertices $v_p(i,j,1)$ and $v_p(i,j,2)$ encode the vertex $(i,j) \in V(H)$; recall that we call such a pair \emph{homologous}.
  We attach to each column $C_{p,j}$, for $p \in [m]$ and $j \in [k]$, a column selector gadget (for $\Pi$), with additional vertices $\mathcal C_{\text{sel}}(p,j)$.
  For each pair $p \in [m], i \in [k]$, we add a row selector gadget to $R_{p,i}$, with additional vertices $\mathcal R_{\text{sel}}(p,i)$.
  
  For each edge $e_p=(i_p,j_p)(i'_p,j'_p) \in E(H)$ ($p \in [m]$), we attach an edge gadget, with additional vertices $\mathcal E_p(i_p,j_p,i'_p,j'_p)$, to $\{v_p(i_p,j_p,1),v_p(i_p,j_p,2),v_p(i'_p,j'_p,1),v_p(i'_p,j'_p,2)\}$.
  For each $p \in [m-1]$, we add a propagation gadget between $H_p$ and $H_{p+1}$, with additional vertices $\mathcal P_p$.
  This finishes the construction of~$G$.
  We set $k' := 2(k-1)km$.
  
  \subparagraph*{Correctness.}
We first assume that there is a solution $I$ to \textsc{$k \times k$-Permutation Independent Set}.
That is, $I$ is an independent set of $H$ with exactly one vertex per column and per row.
Say the vertices of $I$ are $(1,j_1), (2,j_2), \ldots (k,j_k)$ with $\{j_1,j_2,\ldots,j_k\}=[k]$. Then
\[X := \bigcup_{p\in[m]} H_p \setminus \cup_{i \in [k]} \{v_p(i,j_i,1),v_p(i,j_i,2)\}\] 
is a solution to $\Pi$.
Indeed it is $\Pi$-legal since it intersects every edge gadget (if not, the edge gadget would be between two vertices of $I$, a contradiction) and $\Pi$ satisfies the intended-solution property, by assumption.
Furthermore $|X|=2mk(k-1)=k'$.

We now assume that the $\Pi$-instance $(G,k')$ admits a solution (of size $k'$), say $X$.
The graph $G$ has $km$ disjoint $\Pi$-obstructions $C_{p,j} \cup \mathcal C_{\text{sel}}(p,j)$.
For each of these sets, at least $s := 2(k-1)$ vertices must be deleted, by the specification of the column sector gadget.
Since globally only $k'=kms$ vertices can be deleted, $X$ intersects each $C_{p,j} \cup \mathcal C_{\text{sel}}(p,j)$ at a set $C_{p,j} \setminus \{v_p(i_{j,p},j,1),v_p(i_{j,p},j,2)\}$ for some $i_{j,p} \in [k]$.
Moreover, the $k$ row selector gadgets attached to each $H_p$ enforce that $\{i_{1,p},i_{2,p},\ldots,i_{k,p}\}=[k]$, and the propagation gadget $\mathcal P_p$ enforces that $i_{j,p}=i_{j,p+1}$ for every $j \in [k]$.
This implies that $i_{j,1}=i_{j,2}= \ldots = i_{j,m}$ for every $j \in [k]$, and we simply denote this common value by $i_j$.
We claim that $\{(i_1,1), (i_2,2), \ldots, (i_k,k)\}$ is a solution to the \textsc{$k \times k$-Permutation Independent Set} instance.
We have already argued that $\{i_1,i_2,\ldots,i_k\}=[k]$.
Finally there cannot be an edge $e_p = (i_j,j)(i_{j'},j') \in E(H)$ since then the $\Pi$-obstruction $\mathcal E_p(i_j,j,i_{j'},j') \cup \{v_p(i_j,j,1),v_p(i_j,j,2),v_p(i_{j'},j',1),v_p(i_{j'},j',2)\}$ would be disjoint from $X$.

\subparagraph*{Pathwidth in $\Oh(k)$.}
Let $\mathcal P'_p$ be the $\Oh(k)$ vertices of $\mathcal P_p$ with strictly more than one neighbor in $H_p \cup H_{p+1}$.
For every $p \in [m-1]$, we set $Y_p := \mathcal P'_p \cup \mathcal E_p(i_p,j_p,i'_p,j'_p) \cup C_{p,j_p} \cup \mathcal C_{\text{sel}}(p,j_p) \cup C_{p,j'_p} \cup \mathcal C_{\text{sel}}(p,j'_p) \cup \bigcup_{i \in [k]} \mathcal R_{\text{sel}}(p,i)$, and we observe that $|Y_p|=\Oh(k)$ (this is where it is important that each $\mathcal R_{\text{sel}}(p,i)$ has constant size).
For each $p \in [m]$ and $j \in [k-2]$, let $Z_{p,j}$ be $C_{p,j^*} \cup C_{\text{sel}}(p,j^*)$ where $j^*$ is the $j$-th index, by increasing value, in $[k] \setminus \{j_p,j'_p\}$.
Again we notice that $|Z_{p,j}|=\Oh(k)$.

Here is a path-decomposition of $G$ of width $\Oh(k)$ in case every $\mathcal P_p \setminus \mathcal P'_p$ is empty: $Y_1, Y_1 \cup Z_{1,1}, Y_1 \cup Z_{1,2}, \ldots, Y_1 \cup Z_{1,k-2}, Y_1 \cup Y_2, Y_1 \cup Y_2 \cup Z_{2,1}, Y_1 \cup Y_2 \cup Z_{2,2}, \ldots, Y_1 \cup Y_2 \cup Z_{2,k-2}, Y_2 \cup Y_3, \ldots, Y_{p-2} \cup Y_{p-1}, Y_{p-2} \cup Y_{p-1} \cup Z_{p-1,1}, Y_{p-2} \cup Y_{p-1} \cup Z_{p-1,2}, \ldots, Y_{p-2} \cup Y_{p-1} \cup Z_{p-1,k-2}, Y_{p-1}, Y_{p-1} \cup Z_{p,1}, Y_{p-1} \cup Z_{p,2}, \ldots, Y_{p-1} \cup Z_{p,k-2}$.
Indeed the maximum bag size is $\Oh(k)$ and each edge of $G$ appears in at least one bag.
Two crucial properties used in this path-decomposition are that (1) the removal of $\mathcal P'_p \cup \mathcal P'_{p+1}$, so in particular of $Y_p \cup Y_{p+1}$, disconnects $H_{p+1}$ from the rest of $G$, and (2) there is no edge between $Z_{p,j}$ and $Z_{p,j'}$ for $j \neq j' \in [k-2]$ and $p \in [m]$.

In the general case, a path-decomposition of width $\Oh(k)$ for $G$ is obtained from the previous decomposition by observing the following rule.
Each time a vertex of $H_p$ appears in a bag for the first time, we introduce and immediately remove each of its neighbors in $\mathcal P_p \setminus \mathcal P'_p$ one after the other.
\end{proof}

\subsection{Designing ad hoc gadgets}\label{subsec:gadgets}

We now build specific gadgets for \sfvs, \soct, and \ect.
For these problems, we always use $S$ to denote the prescribed subset of vertices through which no cycle, no odd cycle, or no even cycle should go, respectively.

\subsubsection{Column selector gadgets}\label{sec:column-gadgets}

We begin with the column selector gadget $\mathcal G_1(\mathcal C)$ used for \ssfvs and \ssoct, followed by the gadget $\mathcal G_2(\mathcal C)$ used for \shortect.
The column selector gadget $\mathcal G_1(\mathcal C)$ attached to a column $C_{p,j}$ is defined as follows.
It comprises $3k$ additional vertices.
These $3k$ vertices are all added to $S$, and they form an independent set.
Each of the first $k$~vertices, $d_{p,j}(1,1), \ldots, d_{p,j}(k,1) \in S$, are adjacent to all vertices in $\bigcup_{i \in [k]} \{v_p(i,j,1)\}$, so these vertices induce a biclique.
The next $k$~vertices, $d_{p,j}(1,2), \ldots, d_{p,j}(k,2) \in S$, also twins, are adjacent to all vertices in $\bigcup_{i \in [k]} \{v_p(i,j,2)\}$.
We add $d_{p,j}(1), \ldots, d_{p,j}(i), \ldots,$ $d_{p,j}(k)$ and, for each $i \in [k]$, we link $d_{p,j}(i)$ to all the vertices in $\{v_p(i,j,1)\} \cup \bigcup_{i' \in [k]\setminus \{i\}} \{v_p(i',j,2)\}$.
Finally we make every distinct pair $v_p(i,j,z), v_p(i',j,z')$ adjacent, except if $i=i'$.
See~\cref{fig:column-selector} for an illustration.

\begin{figure}[h!]
    \centering
    \resizebox{250pt}{!}{
    \begin{tikzpicture}
      \def\k{3}
      \def\t{2.3}
      \def\s{2}
      \foreach \i in {1,...,\k}{
        \foreach \z in {1,2}{
          \node[draw,circle,inner sep=0.01cm] (v\i\z) at (\z * \s,\i * \t) {\footnotesize{$v_p(\i,j,\z)$}} ;
          \node[draw,circle,double,double distance=1.35pt,inner sep=0.01cm] (d\i\z) at (5 * \z * \s - 6 * \s,\i * \t) {\footnotesize{$d_{p,j}(\i,\z)$}} ; 
        }
        \node[draw,rectangle,rounded corners,thick,fit=(v\i1)(v\i2)] (v\i) {} ;
      }
      \draw[very thick, blue] (v1) -- (v2) -- (v3) ;
      \draw[very thick, blue] (v1.west) to [bend left=30] (v3.west) ;
      \foreach \i/\j in {1/-1,2/3,3/7}{ 
        \node[draw,circle,double,double distance=1.35pt] (d\i) at (\j,\k * \t + \t) {\footnotesize{$d_{p,j}(\i)$}} ;
      }
      \foreach \z in {1,2}{
        \foreach \i in {1,...,\k}{
          \foreach \j in {1,...,\k}{
            \draw (v\i\z) to [bend left=10] (d\j\z) ;
          }
        }
      }
      \draw[thick, red] (d1) to [bend right=10] (v11) ;
      \draw (d1) to [bend right=14] (v22) ;
      \draw (d1) to [bend left=10] (v32) ;

      \draw[thick, red] (d2) to [bend left=8] (v21) ;
      \draw (d2) to [bend right=13] (v12) ;
      \draw (d2) -- (v32) ;

      \draw[thick, red] (d3) to [bend right=10] (v31) ;
      \draw (d3) to [bend left=8] (v12) ;
      \draw (d3) -- (v22) ;
    \end{tikzpicture}
    }
    \caption{The column selector gadget $\mathcal G_1(\mathcal C)$. Doubly-circled vertices are in $S$. Blue edges linking boxes denote bicliques between the two surrounded vertex sets. The gadget $\mathcal G_2(\mathcal C)$ is obtained by subdividing each red edge once, and adding a false twin to $d_{p,j}(k,1)$ (or equivalently, any $d_{p,j}(i,1)$) and a false twin to $d_{p,j}(k,2)$.}
    \label{fig:column-selector}
\end{figure}
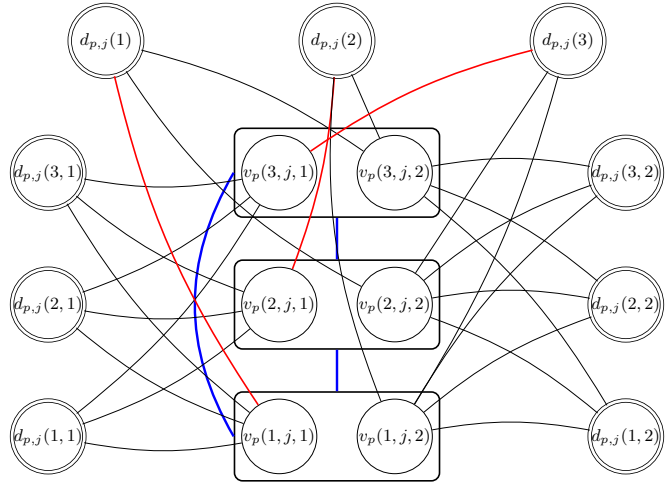

We obtain the column selector gadget $\mathcal G_2(\mathcal C)$ from $\mathcal G_1(\mathcal C)$ by adding, for each $z \in [2]$, a vertex $d_{p,j}(k+1,z)$ adjacent to all vertices in $\bigcup_{i \in [k]} \{v_p(i,j,z)\}$, and by subdividing each edge $d_{p,j}(i)v_p(i,j,1)$ once.

\begin{lemma}\label{lem:column-gadgets}
$\mathcal G_1(\mathcal C)$ is a column selector gadget for \sfvs and \soct, and $\mathcal G_2(\mathcal C)$ is a column selector gadget for \ect.
\end{lemma}
\begin{proof}
  The gadgets $\mathcal G_1(\mathcal C)$ and $\mathcal G_2(\mathcal C)$ add $3k$ and $4k+2$, respectively, new vertices, thus $\Oh(k)$.
  Their edge set respects the specification of the column selector.

  We first show that the only $\Pi$-legal $(2k-2)$-deletions within $\mathcal G_1(\mathcal C)$ are the sets $C_{p,j} \setminus \{v_p(i,j,1),v_p(i,j,2)\}$ (for $i \in [k]$), for $\Pi \in \{$\ssfvs, \textsc{Subset OCT}$\}$.
  For every $p \in [m]$, $j \in [k]$, and $z \in [2]$, the biclique $K_{k,k}$ between $\bigcup_{i \in [k]} \{v_p(i,j,z)\}$ and $\bigcup_{i \in [k]} \{d_{p,j}(i,z)\} \subseteq S$ forces the removal of all but at most one vertex of $\bigcup_{i \in [k]} \{v_p(i,j,z)\}$, or all the vertices in $\bigcup_{i \in [k]} \{d_{p,j}(i,z)\}$.
  Indeed, recall that the former set is a clique, while the latter set is an independent set and is contained in the prescribed set $S$.
  Hence keeping at least one vertex in $\bigcup_{i \in [k]} \{d_{p,j}(i,z)\}$ and at least two in $\bigcup_{i \in [k]} \{v_p(i,j,z)\}$ results in an odd cycle (a triangle) going through at least one vertex of $S$.
  Thus the only $\Pi$-legal $(2k-2)$-deletions within $\mathcal G_1(\mathcal C)$ have to remove exactly $k-1$ vertices in $\bigcup_{i \in [k]} \{v_p(i,j,1)\}$ and exactly $k-1$ vertices in $\bigcup_{i \in [k]} \{v_p(i,j,2)\}$.
  Let $Y$ denote such a deletion set, and observe that $Y \cap S = \emptyset$.
  We further claim that if $v_p(i,j,1)$ is not in $Y$, then $v_p(i,j,2)$ is also not in $Y$.
  Assume, for the sake of contradiction, that $v_p(i,j,1)$ and $v_p(i',j,2)$ are two (adjacent) vertices, not in $Y$, with $i \neq i'$.
  Then $d_{p,j}(i) \in S$ forms a surviving triangle with $v_p(i,j,1)$ and $v_p(i',j,2)$.
  Thus $Y = C_{p,j} \setminus \{v_p(i,j,1),v_p(i,j,2)\}$ for some $i \in [k]$.

  This finishes the proof that $\mathcal G_1(\mathcal C)$ is a column selector gadget for \ssfvs and \ssoct.
  We now adapt the arguments for $\mathcal G_2(\mathcal C)$ and $\Pi =$ \shortect.
  Now the biclique $K_{k,k+1}$ between $\bigcup_{i \in [k]} \{v_p(i,j,z)\}$ and $\bigcup_{i \in [k+1]} \{d_{p,j}(i,z)\} \subseteq S$ forces the removal of all but at most one vertex of $\bigcup_{i \in [k]} \{v_p(i,j,z)\}$, or all but at most one vertex of $\bigcup_{i \in [k+1]} \{d_{p,j}(i,z)\}$, otherwise there would be a surviving even cycle $C_4$.
  Since only $k-1$ vertices can be removed from each $\Pi$-obstruction $\bigcup_{i \in [k]} \{v_p(i,j,z)\} \cup \bigcup_{i \in [k+1]} \{d_{p,j}(i,z)\} \subseteq S$ (with $z \in [2]$), the only $\Pi$-legal $(2k-2)$-deletions within $\mathcal G_2(\mathcal C)$ remove all but one vertex in $\bigcup_{i \in [k]} \{v_p(i,j,1)\}$ and in $\bigcup_{i \in [k]} \{v_p(i,j,2)\}$.
  The end of the proof is similar to the previous paragraph since the triangle $d_{p,j}(i)v_p(i,j,1)v_p(i',j,2)$ is now a $C_4$ (recall that we subdivided the edge $d_{p,j}(i)v_p(i,j,1)$ once).
\end{proof}

\subsubsection{Row selector gadgets}\label{sec:row-gadgets}

The row selector $\mathcal G_1(\mathcal R)$, attached to $R_{p,i}$, consists of two additional vertices $r_1(p,i), r'_1(p,i) \in S$ made adjacent to every vertex in $\bigcup_{j \in [k]} \{v_p(i,j,1)\}$.
The row selector $\mathcal G_2(\mathcal R)$ consists of three additional vertices $r_2(p,i), r'_2(p,i), r''_2(p,i)$, each adjacent to all vertices in $\bigcup_{j \in [k]} \{v_p(i,j,1)\}$.
We put only $r'_2(p,i)$ in $S$, and we add an edge between $r_2(p,i)$ and $r''_2(p,i)$.

\begin{lemma}\label{lem:row-gadgets}
$\mathcal G_1(\mathcal R)$ is a row selector gadget for \sfvs and \ect, and $\mathcal G_2(\mathcal R)$ is a row selector gadget for \soct.
\end{lemma}
\begin{proof}
  The gadgets $\mathcal G_1(\mathcal R)$ and $\mathcal G_2(\mathcal R)$ add $2$ and $3$ new vertices, respectively, thus $\Oh(1)$.
  Their edge set respects the specification of the row selector.

  The set $\{r_1(p,i), r'_1(p,i), v_p(i,j,1), v_p(i,j',1)\}$ is a $\Pi$-obstruction, for every pair $j \neq j' \in [k]$, for every problem \mbox{$\Pi \in \{$\ssfvs, \textsc{ECT}$\}$}.
  Indeed it induces an even cycle (a $C_4$) and, in the case of \ssfvs, we note that this cycle goes through two vertices of~$S$.
  The set $\{r_2(p,i), r'_2(p,i), r''_2(p,i), v_p(i,j,1), v_p(i,j',1)\}$ is a $\Pi$-obstruction, for every pair $j \neq j' \in [k]$, for $\Pi =$ \ssoct.
  Indeed it contains an odd cycle $r_2(p,i)v_p(i,j,1)r'_2(p,i)v_p(i,j',1)r''_2(p,i)$ going through $r'_2(p,i) \in S$.
\end{proof}
Crucially for the intended-solution property, the odd cycle $r_2(p,i)v_p(i,j,1)r''_2(p,i)$ does not contain any vertex of $S$.
  
\subsubsection{Edge gadgets}\label{sec:edge-gadgets}

Let $\mathcal G_1(\mathcal E)$ be the following edge gadget, that we present for $e_p=(i,j)(i',j')$.
We add an edge between $v_p(i,j,1)$ and $v_p(i',j',1)$.
We add a vertex $s_p$ adjacent to both $v_p(i,j,1)$ and $v_p(i',j',1)$.
We add $s_p$ to the set $S \subseteq V(G)$.
The edge gadget $\mathcal G_2(\mathcal E)$ is obtained from $\mathcal G_1(\mathcal E)$ by subdividing the edge $s_pv_p(i',j',1)$ once.

\begin{lemma}\label{lem:edge-gadgets}
$\mathcal G_1(\mathcal E)$ is an edge gadget for \sfvs and \soct, and $\mathcal G_2(\mathcal E)$ is an edge gadget for \ect.
\end{lemma}
\begin{proof}
  Both gadgets introduce a constant number of additional vertices (1 and 2, respectively, so $\Oh(k)$), and their edge set respects the specification.
  The gadget $\mathcal G_1(\mathcal E)$ is an odd cycle (a triangle) with a vertex in $S$, hence an obstruction for \sfvs and \soct.
  The gadget $\mathcal G_2(\mathcal E)$ is an even cycle (a $C_4$), hence an obstruction for \ect.
\end{proof}

\subsubsection{Propagation gadgets}\label{sec:propagation-gadgets}

  We present $\mathcal G_1(\mathcal P)$, a propagation gadget inserted between $H_p$ and $H_{p+1}$.
  We first add an independent set of $2k$ vertices.
  Among them, the $k$ vertices $r_{p,1}, \ldots, r_{p,k}$ represent the row indices in $H_p$ and $H_{p+1}$, while the $k$ other vertices $c_{p,1}, \ldots, c_{p,k}$ represent the column indices.
  We link $r_{p,i}$ to all the vertices in $\bigcup_{j \in [k]} \{v_p(i,j,2)\} \cup \bigcup_{j \in [k]} \{v_{p+1}(i,j,1)\}$.
  Similarly, we link $c_{p,j}$ to all the vertices in $\bigcup_{i \in [k]} \{v_p(i,j,2)\} \cup \bigcup_{i \in [k]} \{v_{p+1}(i,j,1)\}$.
  Finally, we add a vertex $c_p \in S$ adjacent to all the vertices $c_{p,1}, \ldots, c_{p,k}$.

  The gadget $\mathcal G_2(\mathcal P)$ is defined similarly, except that we subdivide the edge $r_{p,i}v_p(i,j,2)$ once, for each $i, j \in [k]$.
  Finally the gadget $\mathcal G_3(\mathcal P)$ adds to $\mathcal G_2(\mathcal P)$, a vertex $c'_{p,j}$, for each $j \in [k]$.
  The vertex $c'_{p,j}$ is linked to $c_{p,j}$ and to $c_p$.

  \begin{lemma}\label{lem:propagation-gadgets}
    $\mathcal G_1(\mathcal P)$ is a column selector gadget for \sfvs, $\mathcal G_2(\mathcal P)$ is a column selector gadget for \soct, and $\mathcal G_3(\mathcal P)$ is a column selector gadget for \ect.
  \end{lemma}
  \begin{proof}
    Let $\mathcal P^1_p := \{r_{p,1}, \ldots, r_{p,k}, c_{p,1}, \ldots, c_{p,k}, c_p\}$.
    The gadget $\mathcal G_1(\mathcal P)$ adds to the base the set $\mathcal P^1_p$ of size $2k+1$, thus $\Oh(k)$.
    Hence it trivially satisfies the technical condition of the propagation gadget.
    The gadget $\mathcal G_2(\mathcal P)$ adds a further $k^2$ vertices, stemming from the subdivision of the edges $r_{p,i}v_p(i,j,2)$.
    These vertices have exactly one neighbor in $H_p \cup H_{p+1}$ and the rest of their neighbors in $\mathcal P^1_p$, so satisfy the specification.
    For the same reason, $\mathcal G_3(\mathcal P)$ also satisfies the specification. 
    We denote by $\mathcal P^2_p$ the set of $2k+1+k^2$ vertices consisting of $\mathcal P^1_p$ plus the subdivision vertices, and $\mathcal P^3_p$ the set of $3k+1+k^2$ vertices added in $\mathcal G_3(\mathcal P)$.
    The edge sets of $\mathcal G_1(\mathcal P), \mathcal G_2(\mathcal P), \mathcal G_3(\mathcal P)$ respect the specification of the propagation selector.

    For every $i,j \neq j' \in [k]$, $\mathcal P^1_p \cup \{v_p(i,j,2),v_{p+1}(i,j',1)\}$ is a $\Pi$-obstruction for $\Pi =$ \ssfvs.
    Indeed $r_{p,i}v_p(i,j,2)c_{p,j}c_jc_{p,j'}v_{p+1}(i,j',1)$ is a cycle (a $C_6$) going through $c_j \in S$.
    Similarly $\mathcal P^2_p \cup \{v_p(i,j,2),v_{p+1}(i,j',1)\}$ is a $\Pi$-obstruction for $\Pi =$ \ssoct, the same cycle being now of odd length (a $C_7$), due to the subdivision of $r_{p,i}v_p(i,j,2)$.
    Finally $\mathcal P^3_p \cup \{v_p(i,j,2),v_{p+1}(i,j',1)\}$ is a $\Pi$-obstruction for $\Pi =$ \shortect since $r_{p,i}w_p(i,j,2)v_p(i,j,2)c_{p,j}c'_{p,j}c_jc_{p,j'}v_{p+1}(i,j',1)$ is an even cycle (a $C_8$), where $w_p(i,j,2)$ is the subdivided vertex stemming from the edge $r_{p,i}v_p(i,j,2)$.
\end{proof}
  
\subsubsection{Wrap-up}\label{sec:wrap-up-hardness}

\begin{theorem}\label{hardness:adhoc}
  Unless the ETH fails, the following problems cannot be solved in time $2^{o(\pw \log \pw)}n^{\Oh(1)}$ on $n$-vertex graphs with pathwidth $\pw$:
  \begin{itemize}
  \item \sfvs,
  \item \soct, and
  \item \ect.
  \end{itemize}
\end{theorem}
\begin{proof}
  We need to check that these problems satisfy the preconditions of \cref{hardness:generic}.
  \Cref{sec:column-gadgets,sec:row-gadgets,sec:edge-gadgets,sec:propagation-gadgets} and \cref{lem:column-gadgets,lem:row-gadgets,lem:edge-gadgets,lem:propagation-gadgets} show how to build the four types of gadgets.
  Which problem uses which version of the gadget is summarized in~\cref{tbl:gadgets}.
  See \cref{fig:subsetFVS} for a schematic representation of the construction for \ssfvs.
  
  \begin{table}[h!]
    \centering
    \begin{tabular}{ccccc}
      \toprule
               & column selector & row selector & edge gadget & propagation gadget \\
      \midrule
      \sfvs   & $\mathcal G_1(\mathcal C)$ & $\mathcal G_1(\mathcal R)$ & $\mathcal G_1(\mathcal E)$ & $\mathcal G_1(\mathcal P)$ \\
      \soct   & $\mathcal G_1(\mathcal C)$ & $\mathcal G_2(\mathcal R)$ & $\mathcal G_1(\mathcal E)$ & $\mathcal G_2(\mathcal P)$ \\
      \ect   & $\mathcal G_2(\mathcal C)$ & $\mathcal G_1(\mathcal R)$ & $\mathcal G_2(\mathcal E)$ & $\mathcal G_3(\mathcal P)$ \\
      \bottomrule\\      
    \end{tabular}
    \caption{The different gadgets used for the different problems.}
    \label{tbl:gadgets}
  \end{table}

  Finally we have to check that the problems have the intended-solution property.
  We shall prove that every set $X := \bigcup_{p \in [m], i \in [k], z \in [2]} \{v_p(i,j_i,z)\}$, with $\{j_1,\ldots,j_k\}=[k]$ and intersecting all the edge gadgets is $\Pi$-legal in any graph $G$ obtained by attaching to the base the four types of gadgets with respect to their specification of \cref{subsec:generic}.
  The set $X$ is a solution to \mbox{$\Pi \in \{$\ssfvs, \ssoct, \textsc{ECT}$\}$}, if and only if no 2-connected component (i.e., a block of size at least~3) of $G-X$ is a $\Pi$-obstruction.
  Indeed no cycle can go through a cut-vertex.

  We first note that there is no 2-connected component within $\mathcal G_1(\mathcal C), \mathcal G_2(\mathcal C), \mathcal G_1(\mathcal R), \mathcal G_1(\mathcal E), \mathcal G_2(\mathcal E)$ restricted to $G-X$.
  For the latter two gadgets, this is because, by assumption, $X$ intersects every edge gadget.
  In a gadget $\mathcal G_2(\mathcal R)$ restricted to $G-X$, there is one 2-connected component, namely a triangle; but none of its vertices belongs to~$S$.

  We now observe that every vertex $c_p$ is a cut-vertex in $\mathcal G_1(\mathcal P)$, $\mathcal G_2(\mathcal P)$, and $\mathcal G_3(\mathcal P)$ restricted to $G-X$.
  So the remaining 2-connected components of $G-X$ are induced cycles $C_4$ of the form $r_{p,i}v_p(i,j,2)c_{p,j}v_{p+1}(i,j,1)$ when $\mathcal G_1(\mathcal P)$ is used, or induced $C_5$ when $\mathcal G_2(\mathcal P)$ is used, or triangle and induced cycle $C_5$ when $\mathcal G_3(\mathcal P)$ is used.
  In the first two cases, none of the vertices of the cycles belongs to~$S$.
  In the third case, no cycle is even.
  This establishes that \ssfvs, \ssoct, and \shortect with their respective combination of gadgets have the intended-solution property. 
  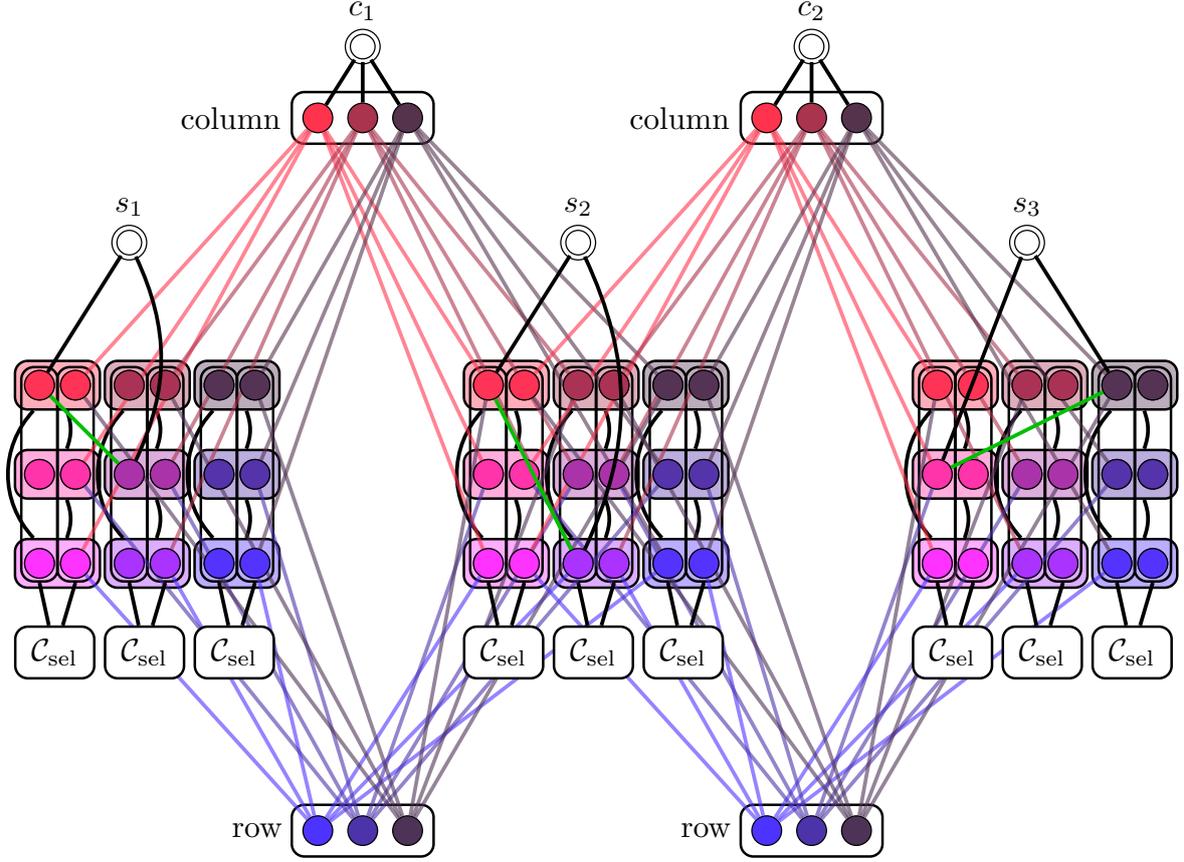
\begin{figure}[h!]
    \centering
    \resizebox{450pt}{!}{
  \begin{tikzpicture}
    \def\k{3}
    \def\s{3}
    \pgfmathtruncatemacro{\sm}{\s - 1}
    \pgfmathtruncatemacro{\kp}{\k + 1}
    \def\vs{5}
    \def\c{1}
    \def\cc{0.5}
    \def\g{3.6}
    \def\y{0.4}
    \def\yy{0.3}
    \foreach \h in {1,...,\s}{
      \foreach \i in {1,...,\k}{
        \foreach \j in {1,...,\k}{
          \node[draw,circle] (v\h\i\j1) at (\c * \i + \vs * \h,\c * \j) {} ;
          \node[draw,circle] (v\h\i\j2) at (\c * \i + \vs * \h+\y,\c * \j) {} ;
          \pgfmathsetmacro{\ii}{1 - \i/\k + 1/\k}
          \pgfmathsetmacro{\jj}{1 - \j/\k + 1/\k}
          \definecolor{colv}{rgb}{\ii,0.2,\jj}
          \node[draw,fill=colv,fill opacity=0.4,rectangle,rounded corners,thick,fit=(v\h\i\j1)(v\h\i\j2),inner sep=0.1cm] (v\h\i\j) {} ;
        }
        \foreach \j in {2,...,\k}{
          \pgfmathtruncatemacro{\jm}{\j - 1}
          \draw[very thick] (v\h\i\j) to [bend left=20] (v\h\i\jm) ; 
        }
        \draw[very thick] (v\h\i1) to [bend left=43] (v\h\i\k) ; 
      }
      \node[draw,circle,double,double distance=1.2pt] (u\h) at (\c * \k / 2 + \c / 2 + \vs * \h,\c * \k + 1.6) {} ;
      \node at (\c * \k / 2 + \c / 2 + \vs * \h,\c * \k + 2) {$s_\h$} ;
    }
    \foreach \h in {1,...,\s}{
      \foreach \j in {1,...,\k}{
        \foreach \z in {1,2}{
          \node[draw,rectangle,thick,rounded corners,fit=(v\h\j1\z)(v\h\j\k\z),inner sep=0.03cm] (column\h\j\z) {} ;
        }
      }
    }
    \foreach \h in {1,...,\sm}{
      \foreach \i in {1,...,\k}{
        \pgfmathsetmacro{\ii}{1 - \i/\k + 1/\k}
        \definecolor{colc}{rgb}{0.3,0.2,\ii}
        \definecolor{colr}{rgb}{\ii,0.2,0.3}
        \node[draw,fill=colc,circle] (r\h\i) at (\cc * \i + \vs * \h + \g,-2) {} ;
        \node[draw,fill=colr,circle] (c\h\i) at (\cc * \i + \vs * \h + \g,\kp * \c +2) {} ;
      }
      \node [draw,rectangle,thick,rounded corners,fit=(r\h1)(r\h\k)] (R\h) {} ;
      \node [draw,rectangle,thick,rounded corners,fit=(c\h1)(c\h\k)] (C\h) {} ;
      \node[left of=R\h] {row~~~~} ;
      \node[left of=C\h] {column~~~~~~~~~} ;
      \node[draw,circle,double,double distance=1.2pt] (uc\h) at (\cc * \k / 2 + \cc / 2 + \vs * \h + \g,\c * \k + 3.8) {} ;
      \node at (\cc * \k / 2 + \cc / 2 + \vs * \h + \g,\c * \k + 4.2) {$c_\h$} ;
    }
    \foreach \h in {1,...,\sm}{
      \foreach \j in {1,...,\k}{
        \draw[very thick] (uc\h) -- (c\h\j) ; 
      }
    }
    \foreach \h in {1,...,\sm}{
      \pgfmathtruncatemacro{\hp}{\h + 1}
      \foreach \i in {1,...,\k}{
        \pgfmathsetmacro{\ii}{1 - \i/\k + 1/\k}
        \definecolor{colr}{rgb}{0.3,0.2,\ii}
        \definecolor{colc}{rgb}{\ii,0.2,0.3}
        \foreach \j in {1,...,\k}{
          \begin{scope}[very thick, opacity=0.6]
          \draw[color=colc] (c\h\i) -- (v\h\i\j2) ;
          \draw[color=colc] (c\h\i) -- (v\hp\i\j1) ;
          \draw[color=colr] (r\h\i) -- (v\h\j\i2) ;
          \draw[color=colr] (r\h\i) -- (v\hp\j\i1) ;
          \end{scope}
        }
      }
    }
    \foreach \h in {1,...,\s}{
      \foreach \i in {1,...,\k}{
        \foreach \j in {1,...,\k}{
          \pgfmathsetmacro{\ii}{1 - \i/\k + 1/\k}
          \pgfmathsetmacro{\jj}{1 - \j/\k + 1/\k}
          \definecolor{colv}{rgb}{\ii,0.2,\jj}
          \node[draw,fill=colv,circle] at (\c * \i + \vs * \h,\c * \j) {} ;
          \node[draw,fill=colv,circle] at (\c * \i + \vs * \h+\y,\c * \j) {} ;
        }
      }
    }
    \foreach \h in {1,...,\s}{
      \foreach \i in {1,...,\k}{
        \node[circle] (sel\h\i1) at (\c * \i + \vs * \h + 0.02,0) {} ;
        \node[circle] (sel\h\i2) at (\c * \i + \vs * \h + 0.02+\yy,0) {} ;
        \node [draw,fill=white,rectangle,thick,rounded corners,fit=(sel\h\i1)(sel\h\i2)] (G\h\i) {} ;
        \node at (\c * \i + \vs * \h + 0.02+\yy / 2,0) {$\mathcal C_{\text{sel}}$} ;
        \draw[very thick] (column\h\i1.south) -- (G\h\i) -- (column\h\i2.south) ;
      }
    }
     \begin{scope}[very thick,black!25!green]
      \draw (v1131) -- (v1221) ;
      \draw (v2131) -- (v2211) ;
      \draw (v3121) -- (v3331) ;
    \end{scope}
    \begin{scope}[very thick]
      \draw (v1131) -- (u1) to [bend left=25] (v1221) ;
      \draw (v2131) -- (u2) to [bend left=25] (v2211) ;
      \draw (v3121) -- (u3) -- (v3331) ;
    \end{scope}
  \end{tikzpicture}
  }
    \caption{Example of the overall picture for \sfvs.
      The first three edges (in green) in the reduction from \textsc{$k \times k$-Permutation Independent Set}, with $k=3$, to \ssfvs.
    The doubly-circled vertices are vertices in $S$.
    The column selector gadget $\mathcal C_{\text{sel}}$, of size $\Oh(k)$, forces that only one pair of homologous vertices is retained in each column (see Figure \ref{fig:column-selector}).
    We did \emph{not} represent the row selector gadget.}
  \label{fig:subsetFVS}
  \end{figure}
  \end{proof}
  
  \subsection{Lower bound for \vmwc}\label{subsec:vmwc-lower}

For \vmwc we will also start from the base $\bigcup_{p \in [m]}H_p$ but we will deviate from the gadget specification of \cref{subsec:generic}.
  We will ``communalize'' the selector, edge, and propagation gadgets.
  That way, we are able to show the claimed lower bound even when the number of terminals is linearly tied to the pathwidth.
  This is unlike our constructions for \ssfvs and \ssoct in \cref{hardness:adhoc} where the size of the prescribed subsets $S$ is significantly larger than the pathwidth.
  
  \begin{theorem}\label{hardness:vmwc}
    Unless the ETH fails, \vmwc cannot be solved in time $2^{o(p \log p)}n^{\Oh(1)}$ on $n$-vertex graphs where $p = \pw + |T|$ is the sum of the pathwidth of the input graph and the number of terminals. 
  \end{theorem}

  \begin{proof}
  We now reduce from \textsc{$k \times k$-Independent Set}.
  Again let $H$ be an $m$-edge \textsc{$k \times k$-Independent Set} instance.
  We build an equivalent \vmwc instance $(G,T,k' := 2(k-1)km)$, with $|T|=k+2$, by adding only $2k+2$ new vertices to the base $\bigcup_{p \in [m]}H_p$.
  We link every non-homologous pair of vertices within each column $C_{p,j}$ (for $p \in [m], j \in [k]$).
  We add two terminals $t, t' \in T$.
  For every edge $e_p=(i_p,j_p)(i'_p,j'_p) \in E(H)$, we make $v_p(i_p,j_p,2)$ and $v_p(i'_p,j'_p,2)$ adjacent.
  We also link $t$ to $v_p(i_p,j_p,2)$, and $t'$ to $v_p(i'_p,j'_p,2)$.

  We add $k$ terminals $r_1, \ldots, r_k \in T$.
  We link every vertex on an $i$-th row ($R_{p,i}$) to $r_i$, except if the vertex is already adjacent to $t$ or $t'$.
  This exception concerns the vertices $v_p(i_p,j_p,2)$ and $v_p(i'_p,j'_p,2)$.
  Finally we add $k$ (non-terminal) vertices $c_1, \ldots, c_k$.
  For each $p \in [m], j \in [k], i \in [k]$, we add an edge between $v_p(i,j,1)$ and $c_i$.
  This finishes the construction of $G$.
  The set of terminals is $T := \{t,t',r_1,\ldots,r_k\}$.
  We ask for a deletion set of size $k' := 2(k-1)km$.
  The pathwidth of $G$ is $\Oh(k)$, since it is obtained by adding $2k+2$ vertices ($\{t,t',r_1,\ldots,r_k\}$) to a graph satisfying the gadget specification of \cref{subsec:generic} (with ``empty'' row selector and propagation gadgets).

  We now show the correctness of this reduction.
  Assume that the graph $H$ admits an independent set $I := \{(i_1,1),(i_2,2),\ldots,(i_k,k)\}$.
  We claim that $X := \bigcup_{p \in [m],j \in [k], z \in [2]} H_m \setminus \{v_p(i_j,j,z)\}$ is a solution to the \vmwc instance.
  We first observe that the connected component of $G-X$ containing~$t$ (and similarly~$t'$) does not contain any other terminal.
  Indeed, since $I$ is an independent set,
  at most one of $v_p(i_p,j_p,2)$ and $v_p(i'_p,j'_p,2)$ is preserved in $G-X$ when $(i_p,j_p)(i'_p,j'_p)$ is an edge of $H$.
  Hence each vertex $v_p(i_p,j_p,2)$ or $v_p(i'_p,j'_p,2)$ that exists in $G-X$ has degree 1: it is adjacent only to $t$ or $t'$.
  So the connected component in $G-X$ of $t$ (resp.~$t'$) is a star centered at $t$ (resp.~$t'$) whose leaves are all in $\bigcup_{p \in [m]} H_p$, hence are non-terminals.
  We now observe that there is no path between $r_i$ and $r_{i'}$ (with $i \neq i'$) in $G-X$.
  Such a path would have to go through a vertex $c_j$.
  Indeed, no edge within a column $C_{p,j}$ is preserved in $G-X$ (nor the edges $v_p(i_p,j_p,2)v_p(i'_p,j'_p,2)$), so there is no other way to go from one row to another.
  But each vertex $c_j$ is adjacent to a single row in $G-X$, since we kept only one pair $v_p(i,j,1), v_p(i,j,2)$ per column $C_{p,j}$, and we made the same choice in every $H_p$.

  Let us now assume that $X$ is a solution to the \vmwc instance $(G,T,k')$.
  A first observation is that no edge within a column $C_{p,j}$ can be present in $G-X$, otherwise there is a $3$-edge path between a pair of terminals in $\{r_1, \ldots, r_k, t,t'\}$, since every edge within $C_{p,j}$ is between non-homologous vertices, and every vertex in $C_{p,j}$ is adjacent to a terminal.
  This implies that for each $p \in [m]$ and $j \in [k]$, we have $\{v_p(i_{p,j},j,1),v_p(i_{p,j},j,2)\} \subseteq C_{p,j} \setminus X$ for some $i_{p,j} \in [k]$.
  In fact, since at least $2(k-1)k$ vertices of $H_p$ must be removed for each $p \in [m]$, and the solution $X$ has size at most $2(k-1)km$, we have $C_{p,j} \setminus X = \{v_p(i_{p,j},j,1),v_p(i_{p,j},j,2)\}$.
  In particular, $X \subseteq \bigcup_{p\in[m]} H_p$, so $c_i \notin X$ for each $i \in [k]$.
  We now show that the $i_{p,j}$'s coincide for each $p\in [m]$.
  Assume for the sake of contradiction that $v_p(i,j,1)$ and $v_{p'}(i',j,1)$ are both present in $G-X$ with $p \neq p'$ and $i \neq i'$.
  Then $r_iv_p(i,j,1)c_jv_{p'}(i',j,1)r_{i'}$ is a path in $G-X$, a contradiction.
  Therefore $i_{1,j}= \ldots = i_{m,j}$.
  Let $i_j$ denote this common value.
  We claim that $\{(i_1,1),\ldots,(i_k,k)\}$ is an independent set in $H$.
  Suppose there is an edge $(i_j,j)(i_{j'},j') \in E(H)$ for distinct $j,j' \in [k]$.
  Then there is a path $tv_p(i_j,j,2)v_p(i_{j'},j',2)t'$ in $G$ for some $p \in [m]$, between the terminals $t$ and $t'$, a contradiction.
  \end{proof}
  
  \subsection{Lower bound for \mwc}\label{subsec:mwc-lower}

  To obtain the lower bound for \mwc, we reduce from \textsc{$k \times k$-Permutation Clique}.

  However, we note that reducing from \textsc{Semi-Regular $k \times k$-Permutation Clique}, where all the vertices of a column have the same degree towards another column, and there is no edge with both endpoints in the same row, would make the construction cleaner.
  So the first reflex is to try and show the same $2^{o(k \log k)}$ lower bound for this variant.
  Ensuring the semi-regularity condition can be done rather smoothly;
  it requires revisiting the grouping technique from, say, \textsc{$3$-Coloring}, and using known results on equitable colorings.
  An interested reader can find a complete proof in the appendix. 
  Nonetheless, getting rid of the ``horizontal'' edges (with both endpoints in the same row) in order to obtain an instance \textsc{$k \times k$-Permutation Clique}, while preserving the semi-regularity, is unnecessarily complex.
  In particular, the reduction from \textsc{$k \times k$-Clique} to \textsc{$k \times k$-Permutation Clique} presented in the seminal paper~\cite{Lokshtanov18} does not preserve semi-regularity.
  To prove the next theorem, we will instead directly reduce from \textsc{$k \times k$-Permutation Clique} and ``regularize'' the degree by some ad hoc gadgetry.
  
  \begin{theorem}\label{hardness:mwc}
    Unless the ETH fails, \mwc cannot be solved in time $2^{o(p \log p)}n^{\Oh(1)}$ on $n$-vertex graphs where $p = \pw + |T|$ is the sum of the pathwidth of the input graph and the number of terminals. 
  \end{theorem}
  
  \begin{proof}
    We reduce from an instance $H$ of \textsc{$k \times k$-Permutation Clique}, 
    so we may assume that there is no edge of $H$ with both endpoints in the same row.
    Let $\mu$ be the number of edges of $H$, and let $\Delta$ be the maximum degree of vertices of $H$.
    We associate each $v \in V(H)$ to the non-negative integer $\delta(v) := \Delta - d_H(v)$, where $d_H(v)$ is the degree of $v$.
    It is useful to consider the graph $H'$ obtained from $H$ by attaching $\delta(v)$ pendant leaves to each $v \in V(H)$, where each vertex in $V(H)$ has degree $\Delta$ in $H'$.
    We set $m := k^2 \Delta$, and observe that $m \geqslant \mu$ corresponds to the number of edges in $H'$.
    
  We build an equivalent \mwc instance $(G,T,k')$, with $|T|=k+1$, by adding a polynomial number of vertices to the base $\bigcup_{p \in [\mu+k^2]}H_p$.
  We do not need the vertices $v_p(i,j,2)$, so we rename every $v_p(i,j,1)$ into simply $v_p(i,j)$.
  Now $R_{p,i}$ is the set $\{v_p(i,1), \ldots, v_p(i,k)\}$ and $C_{p,j}$ is $\{v_p(1,j), \ldots, v_p(k,j)\}$.
  
  We encode weighted edges in the following way.
  When we say that we add an edge of weight $w \in \mathbb N$ between two vertices $u, v$, we mean that we add $w$ ``parallel'' 2-edge paths between $u$ and $v$.
  None of the introduced vertices are terminals, so the instance behaves equivalently as with the weighted edge.
  Thus, for the sake of simplicity, we will treat these parallel paths as a weighted edge.
  If we require an edge to be ``undeletable'', we give it weight $k'+1$, just above the total budget.
  All the weights of the construction are encoded with polynomially many vertices and unit-weight edges.
  Therefore the lower bound does apply to the unweighted version of \mwc.

  We set $h := 12m - k\Delta - {k \choose 2}$ and $k' := (h+1)(k-1)k(\mu+k^2) + h$.
  We add $k$ terminals $r_1, \ldots, r_k \in T$, and we link every vertex in the $i$-th row (that is, in a set $R_{p,i}$) to $r_i$ by an edge of weight $h+1$.
  We add $k$ non-terminals $c_1, \ldots, c_k$ and for each $p \in [\mu+k^2], j \in [k], i \in [k]$, we add an edge of weight $k'+1$ (an undeletable edge) between $v_p(i,j)$ and $c_i$.  
  For every $e_p=(i_p,j_p)(i'_p,j'_p) \in E(H)$ with $p \in [\mu]$, we add the following edge gadget between $v_p(i_p,j_p)$ and $v_p(i'_p,j'_p)$.
  We first build a 5-vertex path where the edge weights are, from one endpoint to the other, $5, 3, 3, 5$.
  We denote by $z_p$ the central vertex, and we link the second vertex, $x_p$, and the fourth vertex, $y_p$, by an edge of weight 3.
  We link the first vertex (one endpoint) of the gadget to $v_p(i_p,j_p)$ and to $r_{i_p}$ by edges of weight 3, and the last vertex (the other endpoint) to $v_p(i'_p,j'_p)$ and to $r_{i'_p}$ by edges of weight 3.
  Finally we link $z_p$ to an additional terminal $t$ (common to every $p \in [\mu]$) by an edge of weight $k'+1$.
  See \cref{fig:edge-gadget-mwc} for an illustration of the edge gadget and how it is attached to the terminals.
  \begin{figure}[h!]
    \centering
    \begin{tikzpicture}
      \node[draw,circle,inner sep=0.02cm] (c) at (0,0) {$z_p$} ;
      \node[draw,circle,inner sep=0.02cm] (cr) at (1,0) {$y_p$} ;
      \node[draw,circle,inner sep=0.02cm] (cl) at (-1,0) {$x_p$} ;

      \node[draw,circle] (wij) at (-2.5,-1) {} ;
      \node[draw,circle] (wijp) at (2.5,1) {} ;
      \node[draw,rounded corners] (vij) at (-4.25,-1) {$v_p(i_p,j_p)$} ;
      \node[draw,rounded corners] (vijp) at (4.25,1) {$v_p(i'_p,j'_p)$} ;

      \draw (cl) to node[midway, below] {3} (c) ;
      \draw (c) to node[midway, below] {3} (cr) ;
      \draw (cl) to [bend left = 30] node[midway, above] {3} (cr) ;

      \draw (cl) to node[midway, below] {5} (wij) ;
      \draw (cr) to node[midway, below] {5} (wijp) ;

      \draw (vij) to node[midway, below] {3} (wij) ;
      \draw (vijp) to node[midway, below] {3} (wijp) ;

      \node[draw,double,double distance=1.35pt,circle,inner sep=0.05cm] (ri) at (-4,-3) {$r_{i_p}$} ;
      \node[draw,double,double distance=1.35pt,circle,inner sep=0.05cm] (rip) at (4,-3) {$r_{i'_p}$} ;
      \node[draw,double,double distance=1.35pt,circle] (t) at (0,-3) {$t$} ;

      \draw (vij) to node[midway, left] {$h+1$} (ri) ;
      \draw (wij) to node[midway, right] {3} (ri) ;
      \draw (vijp) to node[midway, right] {$h+1$} (rip) ;
      \draw (wijp) to node[midway, left] {3} (rip) ;
      \draw (c) to node[midway, right] {$k'+1$} (t) ;
    \end{tikzpicture}
    \caption{The edge gadget for \mwc.}
    \label{fig:edge-gadget-mwc}
  \end{figure}
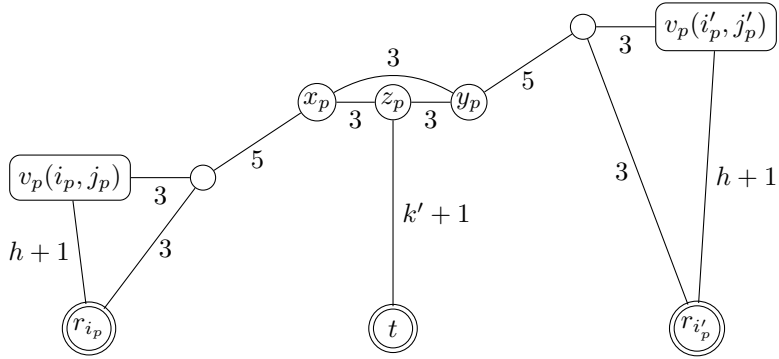
  So far we have added edge gadgets to the first $\mu$ copies $H_1, \ldots, H_\mu$.
  We now describe what we (potentially) add to the last $k^2$ copies $H_{\mu+1}, \ldots, H_{\mu+k^2}$.
  We put an arbitrary total order over $V(H)$, say, the natural $\leqslant$ where $(i,j)$ is interpreted as $n(i,j)=i+(j-1)k$.
  We attach to $v_{\mu+n(i,j)}(i,j)$, $r_i$, and $t$ the following simple gadget, called a \emph{degree-equalizer}, which can be seen as a degenerate case of an edge gadget with multiplicity $\delta((i,j))$ (henceforth we simply write $\delta(i,j)$ for the sake of legibility).
  We add a vertex $w_{\mu+n(i,j)}(i,j)$, and link it to $t$ by an edge of weight $11 \delta(i,j)$, and to $v_{\mu+n(i,j)}(i,j)$ and $r_i$ by edges of weight $6 \delta(i,j)$ each.
  This finishes the construction of $(G,T := \{r_1, \ldots, r_k,t\},k' := (h+1)(k-1)k(\mu+k^2) + h)$.  

  The pathwidth of $G$ is $\Oh(k)$ following the arguments for the \vmwc construction.

  We now show the correctness of the reduction.
  Assume that there is a clique $C := \{(i_1,1), \ldots, (i_k,k)\}$ in $H$, with $\{i_1,\ldots,i_k\}=[k]$.
  We build the following edge deletion-set $X$ for the \mwc instance.
  We start by including in $X$ all the edges of weight $h+1$ between $r_i$ and $v_p(i,j)$ ($p \in [\mu+k^2]$) such that $(i,j) \notin C$.
  This represents $k(k-1)(\mu+k^2)$ weighted edges, and $(h+1)k(k-1)(\mu+k^2)$ unit-weight edges.

  We distinguish three cases for the edge gadget of every $e_p=(i_p,j_p)(i'_p,j'_p)$ ($p \in [\mu]$).
  If $\{(i_p,j_p),(i'_p,j'_p)\} \cap C = \emptyset$ (i.e., $e_p$ has no endpoint in $C$), we add to $X$ the four weight-3 edges incident to $v_p(i_p,j_p)$, $r_{i_p}$, $v_p(i'_p,j'_p)$, and $r_{i'_p}$; a total of 12 edges.
  If $|\{(i_p,j_p),(i'_p,j'_p)\} \cap C| = 1$, say, without loss of generality, that $(i_p,j_p) \in C$, then we add the weight-5 edge incident to $x_p$ and the two weight-3 edges incident to $v_p(i'_p,j'_p)$ and $r_{i'_p}$.
  This consists of 11 edges in total.
  In the symmetric case $(i'_p,j'_p) \in C$, we would remove the weight-5 edge incident to $y_p$ and the two weight-3 edges incident to $v_p(i_p,j_p)$ and $r_{i_p}$.
  Finally if $|\{(i_p,j_p),(i'_p,j'_p)\} \cap C| = 2$, we add the 9 edges of the weighted triangle $x_py_pz_p$ to $X$.

  For every degree-equalizer gadget attached to $H_{\mu+n(i,j)}$, we add to $X$ the weight-$11 \delta(i,j)$ edge incident to $t$ if $(i,j) \in C$, and the two weight-$6 \delta(i,j)$ edges incident to $w_{\mu+n(i,j)}(i,j)$ if $(i,j) \notin C$.
  Note that these numbers of edges correspond to what we would remove in $\delta(i,j)$ copies of an edge gadget where the other endpoint is not in $C$.
  This finishes the construction of $X$.
  
  There are ${k \choose 2}$ edges of $H'$ with both endpoints in $C$, there are $k \Delta - 2 {k \choose 2}$ edges with exactly one endpoint in $C$, and $m-k \Delta+{k \choose 2}$ edges with no endpoint in $C$.
  So there are $9{k \choose 2}+11(k \Delta - 2 {k \choose 2})+12(m-k \Delta+{k \choose 2})=12m-k\Delta-{k \choose 2}=h$ edges added to $X$ from edge and degree-equalizer gadgets.
  Thus $X$ has size $(h+1)k(k-1)(\mu+k^2)+h = k'$ as imposed.
  Let $G'$ be the graph $(V(G),E(G) \setminus X)$.
  We show that every connected component of $G'$ contains at most one terminal.
  Observe that in $G' - \{t\}$, each vertex $z_p$ is in a connected component contained in the edge gadget of $e_p$ (and, in particular, not containing a terminal).
  Since $t$ is only adjacent (by weighted edges) to the vertices $z_p$ and $w_{\mu+n(i,j)}(i,j)$, it follows that the connected component in $G'$ containing $t$ has no other terminals.
  Note furthermore that the removal of the edges in $X$ disconnects every pair $v_p(i_p,j_p), v_p(i'_p,j'_p)$ in the edge gadget of $e_p = (i_p,j_p)(i'_p,j'_p)$ for $p \in [\mu]$.
  Thus the vertices reachable from $r_i$ in $G'$ are $\{c_j\} \cup \bigcup_{p \in [m]} C_{p,j}$, such that $j$ is unique integer of $[k]$ with $i_j=i$, as well as some non-terminal vertices in some edge and degree-equalizer gadgets.
  In particular there is no path between $r_i$ and $r_{i'}$, with $i \neq i'$, in $G'$.
  Thus $X$ is a solution.

  Let us now assume that the $\mwc$ instance $(G,T,k')$ has a solution $X$, and let $G'$ be $(V(G),E(G) \setminus X)$.
  A first observation is that there is a path in $G'$ between any pair of vertices in the $j$-th column, say $v_p(i,j)$ and $v_{p'}(i',j)$, since there are undeletable edges between $c_j$ and each vertex $v_p(i,j)$. 
  Thus there is a component of $G'$ containing $\bigcup_{p \in [\mu+k^2]} C_{p,j}$, for each $j \in [k]$, and this component contains at most one terminal.
  With a budget of $(h+1)k(k-1)(\mu+k^2)+h$, one can remove at most $k(k-1)(\mu+k^2)$ edges of weight $h+1$.
  Since no two edges $r_iv_p(i,j)$ and $r_{i'}v_{p}(i',j)$ can remain in $G'$, for distinct $i,i' \in [k]$, $j \in [k]$, and $p \in [\mu+k^2]$, at least $k(k-1)$ edges of weight $h+1$ incident to a vertex in $H_p$ must be in $X$, for each $p \in [\mu+k^2]$, for a total of at least $k(k-1)(\mu+k^2)$ edges of weight $h+1$.
  Now the only possibility is that, for each $j \in [k]$, there exists an $i_j \in [k]$ such that $X$ contains all the edges of weight $h+1$ from $\bigcup_{p \in [\mu+k^2]}C_{p,j}$ to $\{r_1,\ldots,r_k\}$ except those incident to $r_{i_j}$.
  We set $C := \{(i_1,1), \ldots, (i_k,k)\}$, and we will now show that $C$ is a clique in $H$.
  In particular $\{i_1, \ldots, i_k\}=[k]$ since there is no edge of $H$ with endpoints in the same row.

  First we consider an edge $e_p=(i_p,j_p)(i'_p,j'_p) \in E(H)$ such that $\{(i_p,j_p),(i'_p,j'_p)\} \cap C = \emptyset$.
  Note that, in this case, $v_p(i_p,j_p)$ (resp.~$v_p(i'_p,j'_p)$) is, in $G'$, in the connected component of $r_{i_{j_p}} \neq r_{i_p}$ (resp.~$r_{i_{j'_p}} \neq r_{i'_p}$).
  We then need to separate the seven pairs: $(r_{i_p},v_p(i_p,j_p))$, $(t,v_p(i_p,j_p))$, $(r_{i_p},t)$, $(r_{i'_p},v_p(i'_p,j'_p))$, $(t,v_p(i'_p,j'_p))$, $(r_{i'_p},t)$, and $(r_{i_p},r_{i'_p})$.
  This requires 12 edge deletions.

  We now consider an edge $e_p=(i_p,j_p)(i'_p,j'_p) \in E(H)$ such that $|\{(i_p,j_p),(i'_p,j'_p)\} \cap C| = 1$.
  We assume that $(i_p,j_p) \in C$ (the other case is symmetric).
  In this case, $v_p(i'_p,j'_p)$ is, in $G'$, in the connected component of~$r_{i_{j'_p}} \neq r_{i'_p}$.
  We then need to separate the six pairs: $(t,v_p(i_p,j_p))$, $(r_{i_p},t)$, $(r_{i'_p},v_p(i'_p,j'_p))$, $(t,v_p(i'_p,j'_p))$, $(r_{i'_p},t)$, and $(r_{i_p},r_{i'_p})$.
  This requires 11 edge deletions: the weight-5 edge incident to $x_p$ and the two weight-3 edges incident to $v_p(i'_p,j'_p)$ and to $r_{i'_p}$.

  Finally let us assume that $e_p=(i_p,j_p)(i'_p,j'_p) \in E(H)$ is such that $|\{(i_p,j_p),(i'_p,j'_p)\} \cap C| = 1$.
  Here we need to separate the five pairs: $(t,v_p(i_p,j_p))$, $(r_{i_p},t)$, $(t,v_p(i'_p,j'_p))$, $(r_{i'_p},t)$, and $(r_{i_p},r_{i'_p})$.
  This requires 9 edge deletions: the three weight-3 edges in the triangle $x_py_pz_p$.

  We now turn to the degree-equalizer gadgets.
  If $(i,j) \notin C$, then we need to separate the three pairs $(r_i,v_{\mu+n(i,j)}(i,j))$, $(t,v_{\mu+n(i,j)}(i,j))$, and $(r_i,t)$.
  This requires $12 \delta(i,j)$ edge deletions (the weighted edges $r_iw_{\mu+n(i,j)}(i,j)$ and $v_{\mu+n(i,j)}(i,j)w_{\mu+n(i,j)}(i,j)$).
  If on the contrary $(i,j) \in C$, we only need to separate the two pairs $(t,v_{\mu+n(i,j)}(i,j))$ and $(r_i,t)$.
  This requires $11 \delta(i,j)$ deletions (the weighted edge $tw_{\mu+n(i,j)}(i,j)$).

  We denote by $s$ the number of edges in $H[C]$.
  Since the edge and degree-equalizer gadgets are pairwise edge-disjoint, what we have shown implies that $X$ contains at least $9s+11(k\Delta-2s)+12(m-k\Delta+s)=12m-k\Delta-s$ edges in the edge gadgets.
  As $X$ is of size at most $k'$, we have that $s$ has to be equal to ${k \choose 2}$.
  This implies that $C$ is a clique.
  \end{proof}

By the simple reduction from \mwc to \sfes, given in the introduction, we obtain the following as a corollary.  
\begin{theorem}\label{hardness:sfes}
  Unless the ETH fails, \sfes cannot be solved in time $2^{o(p \log p)}n^{\Oh(1)}$ on $n$-vertex graphs where $p = \pw + |S|$ is the sum of the pathwidth of the input graph and the number of undeletable (terminal) edges.
\end{theorem}

It is not difficult to adapt the construction of Theorem~\ref{hardness:mwc} for the directed variant of \mwc.  
\begin{theorem}\label{hardness:dmwc}
  Unless the ETH fails, \dmwc cannot be solved in time $2^{o(\pw \log \pw)}n^{\Oh(1)}$ on $n$-vertex directed graphs whose underlying undirected graph has pathwidth $\pw$.
\end{theorem}

\section{Slightly superexponential algorithms}\label{sec:algorithms}

In this section, we present $2^{\Oh(\tw \log \tw)}n^3$-time algorithms for the weighted variants of the considered problems with the exception of \shortect.

	We first present in Theorem~\ref{thm:algosoct} a $2^{\Oh(\tw \log \tw)}n^3$-time algorithm for \ssoct.
	Then, we show that with simple modifications this algorithm can solve \ssfvs.
	We deduce the algorithms for the other problems by reducing these problems to the weighted variant of \ssfvs.

	Let us focus on the \ssoct problem. 
	For a graph $G$ and a vertex set $S$ of $G$, we say that $G$ is \emph{$S$-bipartite} if it has no odd cycle containing a vertex of $S$.
	Solving \ssoct is equivalent to find an $S$-bipartite induced subgraph of maximum size.
	The following characterization of $S$-bipartite graphs will be useful.

\begin{lemma}\label{lem:Sbipartite}
	A graph $G$ is $S$-bipartite if and only if 
	for every block $B$ of $G$, either $B$ has no vertex of $S$, or it is bipartite.
\end{lemma}
\begin{proof}		
	($\Rightarrow$)
	Assume toward a contradiction that $G$ is $S$-bipartite and that a block $B$ of $G$ contains a vertex $s\in S$ and $B$ is not bipartite.
	Because $B$ is not bipartite, there exists an odd cycle $C$ in $B$.
	Since $G$ is $S$-bipartite by assumption, $C$ does not contain $s$.
	
	Since $B$ is 2-connected and has at least 3 vertices, there exist two paths $P_{sc}$ and $P_{c's}$ between $s$ and two distinct vertices $c,c'$ of $C$ such that the internal vertices of $P_{sc}$ and $P_{c's}$ and the vertices of $C$ are pairwise distinct.
	Let $P_{cc'}$ and $\widehat{P}_{cc'}$ be the two paths between $c$ and $c'$ in $C$.
	The concatenations $C_1=P_{sc}\cdot P_{cc'} \cdot P_{c's}$ and $C_2=P_{sc}\cdot \widehat{P}_{cc'} \cdot P_{c's}$ are two $S$-traversing cycles.
	Since $C$ is an odd cycle, the parity of $P_{cc'}$ and $\widehat{P}_{cc'}$ are not the same.
	Hence, one of the two cycles $C_1$ and $C_2$ is an odd $S$-traversing cycle. This yields a contradiction.
	
	\medskip
	
	($\Leftarrow$) Assume that $G$ is not $S$-bipartite.
	Then, $G$ contains an odd $S$-traversing cycle $C$. This cycle is contained in a block $B$ of $G$.
	Thus $G$ has a block that is not bipartite and that contains at least one vertex in $S$.
\end{proof}

One can easily modify the proof of the first direction of Lemma \ref{lem:Sbipartite} to prove the following fact.

\begin{fact}\label{fact:notbipartite}
	If a graph $G$ is 2-connected and not bipartite, then there exists an odd path and an even path between every pair of vertices.
\end{fact}

\begin{theorem}\label{thm:algosoct}
	(\textsc{Weighted}) \soct can be solved in time $2^{\Oh(\tw \log \tw)}n^3$ on $n$-vertex graphs with treewidth $\tw$.
\end{theorem}
\begin{proof}
In the following, we fix a graph $G$, $S\subseteq V(G)$, and a weight function $w:V(G)\rightarrow \mathbb{R}$.
Using Theorem~\ref{thm:approxtw} and Lemma~\ref{lem:nicetd}, 
we obtain a nice tree decomposition of $G$ of width at most $5\tw+4$ in time $\mathcal{O}(c^\tw\cdot n)$ for some constant $c$.
Let $(T, \{B_t\}_{t\in V(T)})$ be the resulting nice tree decomposition.
For each node $t$ of $T$, let $G_t$ be the subgraph of $G$ induced by the union of all bags $B_{t'}$ where $t'$ is a descendant of $t$.

\medskip

Let $t$ be a node of $T$.
A \textit{partial solution of $G_t$} is a subset $X\subseteq V(G_t)$ such that $G[X]$ is $S$-bipartite.
In the following, we introduce a notion of auxiliary graph in order to design an equivalence relation  $\equiv_t$  between partial solutions such that $X\equiv_t Y$ if, for every $W\subseteq V(\overline{G_t})$, $G[X\cup W]$ is $S$-bipartite if and only if $G[Y\cup W]$ is $S$-bipartite.

Let $X\subseteq V(G)$ (not necessarily contained in $G_t$).
We denote by $\inc(X)$ the block-cut tree of $G[X]$, that is the bipartite graph whose vertices are the blocks and the cut vertices of $G[X]$ and where a block~$B$ is adjacent to a cut vertex $v$ if $v\in V(B)$.
Observe that $\inc(X)$ is by definition a forest.

We say that a vertex $v$ of $\inc(X)$ is \textit{active} (with respect to $t$) if:
\begin{itemize}
	\item $v$ is a cut vertex of $G[X]$ in $B_t$,
	
	\item $v$ is a block of $G[X]$ that contains at least two vertices in $B_t$, or
	
	\item $v$ is a block of $G[X]$ that contains exactly one vertex in $B_t$ that is not a cut vertex.
\end{itemize}
Note that every vertex in $B_t$ is an active cut vertex or it is in an active block of $G[X]$.
Intuitively, the auxiliary graph associated with a partial solution $X$ needs to encode how the active blocks of $\inc(X)$ are connected together.

We construct the auxiliary graphs $\aux_p(X,t)$ and $\aux(X,t)$ from $\inc(X)$ by the following operations:
\begin{enumerate}
	\item We remove recursively the leaves and the isolated vertices that are inactive. Let $\aux_p(X,t)$ be the resulting graph ($p$ for ``prototype'').
	
	\item For every maximal path $P$ of $\aux_p(X,t)$ between $u$ and $v$ and with inactive internal vertices of degree~2, we remove the internal vertices of $P$ and we add an edge between $u$ and $v$ (shrinking degree 2 nodes that are inactive).
\end{enumerate}  
Figure \ref{fig:auxiliary} illustrates the constructions of $\aux_p(X,t)$ and $\aux(X,t)$.
Observe that Operation 1 removes the inactive blocks of $G[X]$ that contain one vertex in $B_t$.
Thus, every block in $\aux_p(X,t)$ that contains vertices in $B_t$ is active.
By construction, $\aux(X,t)$ is a forest whose vertices are the active vertices of $\inc(X)$ and the inactive vertices that have degree at least 3 in $\aux_p(X,t)$.
An important remark is that the algorithm uses the graphs $\aux(X,t)$ for $X\subseteq V(G_t)$ and in the proof we will use $\aux(X,t)$ and $\aux_p(X,t)$ for $X\subseteq V(G_t)$ or $X\subseteq B_t \cup V(\overline{G_t})$.

By Step 2, any edge $uv$ of $\aux(X, t)$ corresponds to an alternating sequence $P$ of cut vertices and blocks $A_1, A_2, \ldots, A_x$ that forms a path from $u=A_1$ to $v=A_x$ in $\inc(X)$. 
We define the graph $M_{uv}$ as the union of the blocks in $P$.
Note that one of $A_1$ and $A_2$ is a cut vertex and one of $A_{x-1}$ and $A_x$ is a cut vertex. We say that these cut vertices are the endpoints of $M_{uv}$.

\begin{figure}[h]
	\centering
	\includegraphics[width=0.9\linewidth]{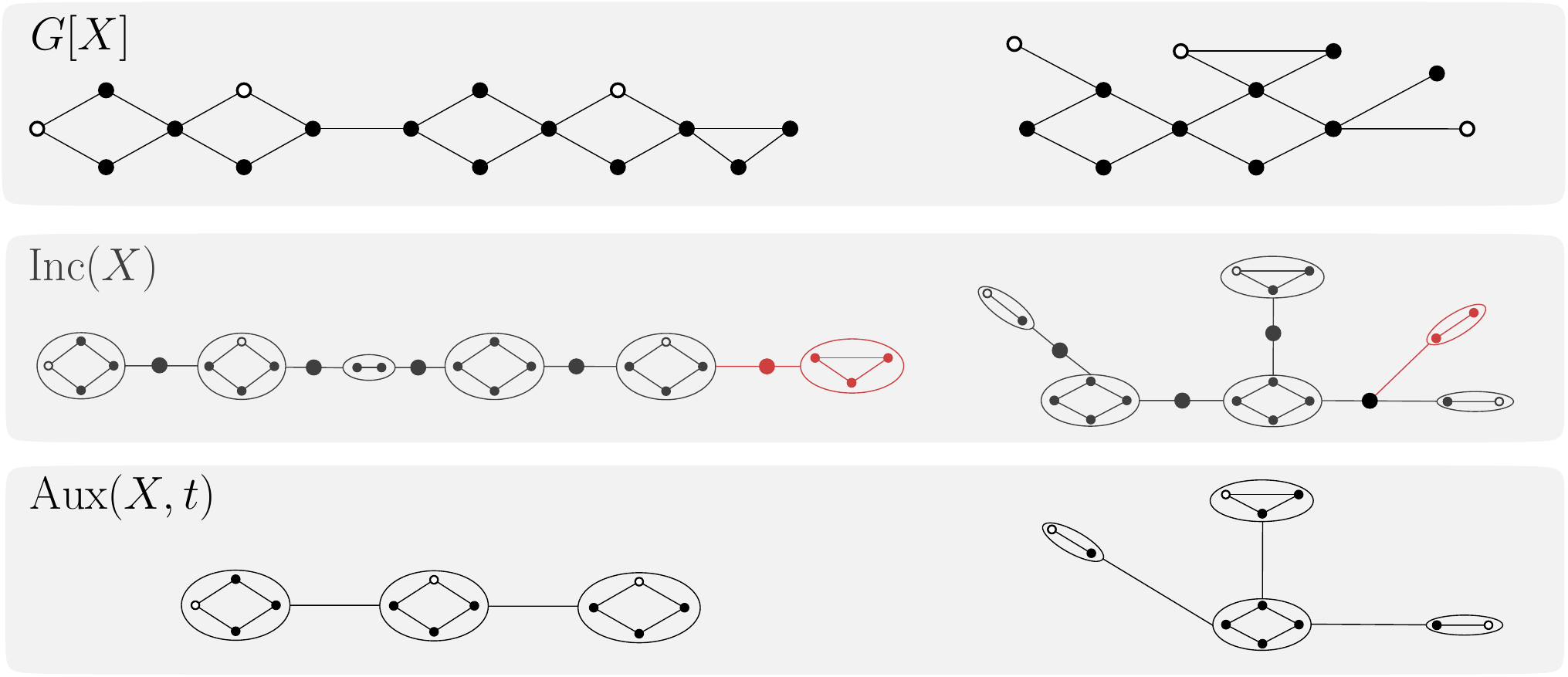}
	\caption{Example of graphs $\inc(X)$ and $\aux(X,t)$ constructed from a graph $G[X]$. The vertices in $B_t$ are white filled.
	The red vertices and edges in $\inc(X)$ are those we remove to obtain $\aux_p(X,t)$.}
	\label{fig:auxiliary}
\end{figure}

\medskip

Let $X$ and $Y$ be two partial solutions of $G_t$.
We say that $X\equiv_t Y$ if $X\cap B_t=Y\cap B_t$, and there is an isomorphism $\varphi$ from $\aux(X,t)$ to $\aux(Y,t)$ such that the following conditions are satisfied:
\begin{enumerate}	
	\item For every vertex $v$ in $\aux(X,t)$, $v$ is active if and only if $\varphi(v)$ is active.
	\item For every vertex $v$ in $\aux(X,t)$, $v$ is a block if and only if $\varphi(v)$ is a block.
	\item For every active cut vertex $v$ in $\aux(X,t)$, we have $\varphi(v)=v$.
	\item For every active block $B$ in $\aux(X,t)$:
	\begin{enumerate}
		\item $V(B)\cap B_t = V(\varphi(B))\cap B_t$, 
		\item $V(B)\cap S \neq \emptyset$ if and only if $V(\varphi(B))\cap S\neq \emptyset$, and
		\item $B$ is bipartite if and only if $\varphi(B)$ is bipartite.
	\end{enumerate}
	\item For every edge $uv$ in $\aux(X,t)$:
	\begin{enumerate}
		\item $M_{uv}$ is bipartite if and only if $M_{\varphi(u)\varphi(v)}$ is bipartite, and
		\item $V(M_{uv})\cap S\neq \emptyset$ if and only if $V(M_{\varphi(u)\varphi(v)})\cap S\neq \emptyset$.
	\end{enumerate} 
	\item For every pair $(u,v)$ of vertices in $B_t\cap X$ and every path $P_X$ between $u$ and $v$ in $G[X]$, there exists a path $P_Y$ in $G[Y]$ between $u$ and $v$ with the same parity as $P_X$.
\end{enumerate}

\begin{claim}\label{claim:boundequiv}
	For every node $t$ of $T$, the equivalence relation $\equiv_t$ has $2^{\Oh(\tw \log \tw)}$ equivalence classes.
\end{claim}
\begin{proof}
	Let $t$ be a node and $X$ be a partial solution of $G_t$. Let $k=\abs{B_t}$.
	In the following, we will upper bound the number of possibilities for the conditions in the definition of $\equiv_t$.
	Notice that there are at most $2^k$ possibilities for $X\cap B_t$.
	
	Now, observe that the number of active blocks of $G[X]$ is at most $k$.
	Note that if an active block contains one vertex of $B_t$, then it is not a cut vertex of $G[X]$, and if an active block intersects at least two vertices of $B_t$, then it contains either two cut vertices of $G[X]$ contained in $B_t$, or it contains at least one vertex of $B_t$ that is not a cut vertex of $G[X]$. We consider $\inc(X)$ as a rooted forest (where each tree has a root), and
	give an injection $\phi$ from the set of active blocks to $B_t$ as follows.
	For each active block $B$, 
	\begin{itemize}
		\item if it contains a vertex in $B_t$ that is not a cut vertex of $G[X]$, 
		then choose such a vertex $v$ and set $\phi(B)=v$, and
		
		\item if all vertices of $B_t\cap V(B)$ are cut vertices of $G[X]$, then $\abs{B_t\cap V(B)} \geq 2$ and we choose one vertex $v\in B_t\cap V(B)$ that is a child of $B$ in $\inc(X)$, and set $\phi(B)=v$.

	\end{itemize}
	Clearly, $\phi$ is an injection, and it shows that the number of active blocks of $G[X]$ is at most $k$.
	We deduce that $\aux(X,t)$ contains at most $2k$ active vertices as there are at most $k$ active blocks and at most $k$ active cut vertices.

	By construction, all the vertices of degree at most 2 in $\aux(X,t)$ are active vertices of $\inc(X)$. 
	In particular, the leaves of $\aux(X,t)$ are active vertices.
	The leaves connected to vertices of degree at least~3 induce an independent set in $\aux(X,t)$.
	We can easily show from  this fact that there are at most $k$ leaves connected to vertices of degree 3 (at most the number of part in a partition of $k$ elements).
	Since $\aux(X,t)$ is a forest, the number of vertices of degree at least 3 is at most $k$.
	We deduce that $\aux(X,t)$ has at most $2k$ active vertices and $k$ inactive vertices.
	Hence, there are at most $(2k+1)(k+1)$ possibilities for Condition~1 as it is at most the number of ways of choosing two numbers one in $[0,2k]$ and one in $[0,k]$.
	By Cayley's formula \cite{Cayley89}, the number of forests on $3k$ labeled vertices is $(3k+1)^{3k-1}$.
	Thus, there are $2^{\Oh(\tw\log \tw)}$ non-isomorphic graphs in $\{\aux(W,t)\mid W$ is a partial solution of $G_t\}$.

	For Condition 2, there are 2 possibilities for each vertex in $\aux(X,t)$: either it is a block or a cut vertex.
	Thus, there are at most $2^{3k}$ possibilities for this condition.
	
	We claim that there are $2^{\Oh(k\log k)}$ possibilities for Conditions 3 and 4.a.
	Let $v_1,\dots,v_d$ be the cut vertices of $G[X]$ in $B_t$ and $X_1,\dots,X_\ell$ be the intersections between $B_t$ and the vertex sets of the active blocks of $G[X]$.
	Note that for every distinct $X_i$ and $X_j$, $\abs{X_i\cap X_j}\le 1$.
	Moreover, since they came from $\inc(X)$, there is no cyclic structure; that is, $X_{i_1}-v_{i_2}-X_{i_3}\cdots -v_{i_{\alpha-1}}-X_{i_{\alpha}}$ where $X_{i_1}=X_{i_{\alpha}}$ and each $v_{i_j}$ only belongs to $X_{i_j}$ and $X_{i_{j+1}}$.
	This means that the number of possibilities for  $v_1,\dots,v_d$ and $X_1, \ldots, X_{\ell}$
	is the same as the number of ways of partitioning a set of $k$ vertices into blocks and cut vertices, as isolated vertices are single blocks.
	
	We claim that the number of ways of partitioning a set of $k$ vertices into blocks and cut vertices is $2^{\mathcal{O}(k\log k)}$.
	Let $T$ be a set of $k$ vertices. 
	First take a partition $\cP$ of $T$. There are at most $2^{k \log k}$ possibilities for $\cP$.
	Choose among the singletons of $\cP$ the cut vertices. There are at most $2^{k}$ possibilities.
	The other parts of $\cP$ indicate the vertex set of blocks after removing cut vertices.
	We add $k$ new dummy parts to $\cP$ representing the possible blocks that may contain only cut vertices.
	Now, we have at most $2k$ parts in $\cP$.
	Observe that any forest between the parts of $\cP$ that represent the cut vertices and those that represent the blocks induces one way of decomposing the $k$ vertices into blocks and cut vertices.
	A dummy part adjacent to the singletons containing the vertices $w_1, w_2, \ldots, w_{\ell}$ indicate that $\{w_1, w_2, \ldots, w_{\ell}\}$ forms a block.
	By Cayley's formula~\cite{Cayley89}, the number of forests on $r$ labeled vertices is $(r+1)^{r-1}$.
	So, there are at most $(2k+1)^{2k-1}$ ways.
	Hence, the number of ways of partitioning a set of $k$ vertices into blocks and cut vertices is at most $2^{\mathcal{O}(k\log k)}$.
	
	For Conditions 4.b and 4.c, there are 3 possibilities for each active block of $\aux(X,t)$.
	Indeed, since $G[X]$ is $S$-bipartite and by Lemma~\ref{lem:Sbipartite}, if a block is not bipartite, then it cannot contain vertices in $S$. Thus, there are at most $3^k$ possibilities for Condition 4.b and 4.c.
	
	For Condition 5, there are 6 possibilities for each edge $uv$ of $\aux(X,t)$: 3 possibilities for the parity of paths between the endpoints of $M_{uv}$ and two for the existence of a vertex in $S$ in $M_{uv}$.
	Since $\aux(X,t)$ is a forest with at most $3k$ vertices, we have at most $6^{3k-1}$ possibilities for Condition 5.
	
	It remains to upper bound the number of possibilities for Condition 6 on the parities of the paths between the vertices in $B_t$.
	Let $u$ and $v$ be two vertices in $X\cap B_t$.
	If there is no path between $u$ and $v$, then we can see this in $\aux(X,t)$ as it implies that there is no path between the active vertices associated with $u$ and $v$.
	
	Assume that $u$ and $v$ are connected in $G[X]$.
	Let $P_{uv}$ be a path between $u$ and $v$ in $G[X]$.
	Let $B_u$ be the block of $G[X]$ that contains $u$ and its neighbor in $P_{uv}$.
	Similarly, let $B_v$ be the block that contains $v$ and its neighbor $P_{uv}$.
	Observe that we can have $B_u=B_v$ if and only if $u$ and $v$ are in the same block.
	By construction, there exists a unique path $P$ between $B_u$ and $B_v$ in $\aux(X,t)$ (this path can have length 0 if $B_u=B_v$).
	If there exists a block $B$ in $P$ that is not bipartite, then by Fact \ref{fact:notbipartite}, we deduce that there exist an odd path and an even path between $u$ and $v$. 
	If such a non-bipartite block $B$ exists, then either $B=B_u=B_v$ or there exists an edge $uv$ used by $P$ such that $M_{uv}$ contains $B$.
	If $B=B_u=B_v$, then $B$ is an active block of $\inc(X)$ since it contains at least two vertices in $B_t$. 
	In this case, Condition 4.b stores the information that $B$ is not bipartite.
	Otherwise, if there is an edge $uv$ used by $P$ such that $M_{uv}$ contains $B$, then $M_{uv}$ is not bipartite and Condition 5 stores this information.
	
	Suppose now that every block in $P$ is bipartite.
	Let $H$ be the subgraph that is the union of the bipartite blocks in $G[X]$, and let $(X_1,X_2)$ be a bipartition of $H$.
	By assumption the union of the blocks in $P$ is a subgraph of $H$.
	Thus, every path between $u$ and $v$ has the same parity and this only depends on whether $u$ and $v$ belong to the same part of $(X_1 \cap B_t,X_2 \cap B_t)$.
	We deduce that the number of possibilities for Condition 6 are at most $2^{k}$.
	
	For each condition on $\equiv_t$, we proved that the number of possibilities are $2^{\Oh(k)}$ or $2^{\Oh(k \log k)}$.
	Since $k=B_t\leq 5\tw + 4$, we conclude that $\equiv_t$ has at most $2^{\Oh(\tw \log \tw)}$ equivalence classes.
\end{proof}

\begin{claim}
	Let $t$ be a node of $T$ and $X,Y$ be two partial solutions associated with $t$.
	If $X \equiv_t Y$, then, for every $Z\subseteq V(\overline{G_t})$, the graph $G[X\cup Z]$ is $S$-bipartite if and only if $G[Y\cup Z]$ is $S$-bipartite.
\end{claim}
\begin{proof}
	Assume that $X\equiv_t Y$ and let $Z\subseteq V(\overline{G_t})$ such that $G[Y\cup Z]$ is $S$-bipartite.
	Let $\varphi$ be the isomorphism from $\aux(X,t)$ to $\aux(Y,t)$ given by $\equiv_t$.
	We show that $G[X\cup Z]$ is $S$-bipartite.
	This will prove the claim because $\equiv_t$ is an equivalence relation.
	To prove that $G[X\cup Z]$ is $S$-bipartite, by Claim \ref{lem:Sbipartite}, it is sufficient to prove that every block $B$ of $G[X\cup Z]$ is $S$-bipartite.
	
	Let $G_X:=G[X], G_Y:=G[Y]$, and $G_Z:=G[Z\cup (B_t\cap X)]$.
	We observe that 
	$G[X\cup Z]:=(G_X, (B_t\cap X))\oplus (G_Z, (B_t\cap X))$ and $G[Y\cup Z]:=(G_Y, (B_t\cap X))\oplus (G_Z, (B_t\cap X))$.
	
	Let $B$ be a block of $G[X\cup Z]$.
	Observe that if $B$ is a block of $G[X]$, then it is $S$-bipartite because $X$ is a partial solution.
	Moreover, if $B$ is a block of $G[Z]$, then it is $S$-bipartite because $G[Y\cup Z]$ is $S$-bipartite.
	
	In the following, we assume that $B$ contains vertices from $X$ and $Z$.
	Consequently, $B$ has at least 2 vertices in $B_t$.
	Observe that for every block $B'$ of $G_X$ or $G_Z$, either $\abs{V(B')\cap V(B)}\leq 1$, or $B'$ is fully contained in $B$. 
	Thus, all the blocks of $G_X$ or $G_Z$ contained in $B$ are in $\aux_p(X,t)$ or $\aux_p(Z \cup (X \cap B_t),t)$.

	We will take a corresponding $2$-connected subgraph in $G[Y\cup Z]$.
	Let $\cB_X$ be the set that contains the blocks and the cut vertices of $\aux(X,t)$ contained in $B$.
	We take the subset $Y_B$ of $Y$ that contains (1)~$\varphi(v)$ for every cut vertex $v$ in $\cB_Y$, (2)~$V(\varphi(B'))$ for every blocks $B'$ in $\cB_X$ and (3)~$V(M_{\varphi(u)\varphi(v)})$ for every edge $uv$ of $\aux(X, t)$ with $u, v\in \mathcal{B}_X$. 	
	Let $F:=G[Y_B\cup (V(B)\cap Z)]$. 
	Since all the blocks in $\aux(X,t)$ with vertices in $S$ are active, Condition 4.a of $\equiv_t$ guarantees that $V(B)\cap B_t= V(F)\cap B_t$.
	
	We claim that $F$ is a $2$-connected induced subgraph of $G[Y\cup Z]$ such that 
	\begin{itemize}
		\item $F$  contains a vertex of $S$ if and only if $B$ contains a vertex of $S$, 
		\item $F$ is bipartite if and only if $B$ is bipartite.
	\end{itemize}
	By Lemma \ref{lem:Sbipartite}, 
	this will imply that $B$ is $S$-bipartite, because $F$ is a subgraph of the $S$-bipartite graph $G[Y\cup Z]$. 
	Let $F_Y:=F\cap G_Y$ and $F_Z:=F\cap G_Z$.
	\medskip
	
	(1) ($F$ is $2$-connected.)
	It is not difficult to see that $F$ is connected from the construction of $Y_B$ and because $Y\equiv_t X$ implies that (1)~$V(B')\cap B_t = V(\varphi(B'))\cap B_t$ for every active block of $G[X]$ and~(2) two blocks $B',\widehat{B}$ of $\aux(X,t)$ are connected if and only if $\varphi(B')$ and $\varphi(\widehat{B})$ are connected in $\aux(Y,t)$.
	
%
	Assume towards a contradiction that $F$ has a cut vertex $c$.  
	Let $a_1,a_2$ be the neighbors of $c$ that are contained in distinct components of $F-c$. For each $i\in \{1, 2\}$, let $U_i'$ be a block of $G_Y$ or $G_Z$ containing $a_ic$.
	Note that $U_1'$ and $U_2'$ are fully contained in $F$, because every block of $G_Y$ or $G_Z$ containing two vertices of $F$ is contained in $F$.
	Therefore, $U_i'$ appears in $\aux_p(Y, t)$ if it is a block of $G_Y$ and it appears in $\aux_p(V(G_Z), t)$ otherwise.
	
	Now, we choose $U_1$ and $U_2$ in $\aux(Y, t)$ and $\aux(V(G_Z), t)$ related to $U_1'$ and $U_2'$, respectively.
	\begin{itemize}
		\item (Case 1. $U_1', U_2'$ are in the same part of $\aux_p(Y, t)$ or $\aux_p(V(G_Z), t)$.)
		
		Let us assume that both $U_1'$ and $U_2'$ are in $\aux_p(Y,t)$. A similar argument holds for the other case.
		As $c$ is the intersection of $U_1'$ and $U_2'$ which are blocks of $G_Y$, 
		$c$ is a cut vertex of $G_Y$, and it appears in $\aux_p(Y,t)$.
		Following the path of $\aux_p(Y, t)$ with direction from $c$ to $U_i'$, we choose the first vertex $U_i$ in $\aux(Y, t)$.
		\item (Case 2. $U_1', U_2'$ are not in the same part of $\aux_p(Y, t)$ or $\aux_p(V(G_Z), t)$.)
		
		Without loss of generality, we assume that $U_1'$ is in $\aux_p(Y, t)$ and $U_2'$ is in $\aux_p(V(G_Z), t)$.
		We explain how to choose $U_1$. The symmetric argument is applied to $U_2$.
		In this case, $U_1'$ and $U_2'$ share $c$, and therefore, 
		either $c$ is an active cut vertex in $G_Y$, or $U_1'$ is an active block of $G_Y$.
		In the former case, 
		following the path of $\aux_p(Y, t)$ with direction from $c$ to $U_1'$, we choose the first vertex $U_1$ in $\aux(Y, t)$.
		In the latter case, we set $U_1:=U_1'$.
	\end{itemize}

	For each $i\in \{1,2\}$, if $U_i$ is block in $\aux(Y,t)$, let $V_i=\varphi^{-1}(U_i)$, otherwise, if $U_i$ is a block in $\aux(V(G_Z),t)$, let $V_i=U_i$.
	Because of the construction, $V_1$ and $V_2$ are blocks contained in $B$.
	We choose a vertex $c_X$ in $B$ corresponding to $c$ in $F$.
	If $c\in B_t \cup Z$, then we set $c_X = c$.
	Otherwise, $c\in Y \setminus B_t$ and $c$ must be a cut vertex in $G_Y$ and either (A)~$c$ is a vertex of $\aux(Y,t)$ with neighbors $U_1$ and $U_2$ or (B)~$U_1U_2$ is an edge of $\aux(Y,t)$ and $c$ is a vertex of $M_{U_1U_2}$.
	If (A) holds, then we set $c_X = \varphi^{-1}(c)$ and if (B)~holds, then we take $c_X$ a cut vertex in $M_{\varphi^{-1}(U_1)\varphi^{-1}(U_2)}$ (every $M_{uv}$ admits at least one cut vertex by definition).

	Since $B$ is $2$-connected, there is a path $p_1p_2 \cdots p_m$ from $V_1-c_X$ to $V_2-c_X$ in $B-c_X$.
	This provides a sequence $B_1, B_2, \ldots, B_{m'}$ of blocks that appear in $\aux(X, t)$ or $\aux(V(G_Z), t)$
	 such that $B_1=V_1$ and $B_{m'}=V_2$ and for every $i\leq m'-1$, either $B_iB_{i+1}$ is an edge in $\aux(X, t)$ or in $\aux(V(G_Z), t)$, or $B_i$ and $B_{i+1}$ are contained in distinct parts of $G_X$ and $G_Z$ and they share a vertex in $B_t$.
	 
	For every $i\leq m'$, let $\widehat{B}_i$ be $B_i$ if $B_i$ is a block in $\aux(V(G_Z),t)$ or $\varphi(B_i)$ if $B_i$ is a block in $\aux(X,t)$.
	Let $i\leq m'-1$. 
	If $B_i$ and $B_{i+1}$ is an edge of $\aux(X,t)$, then $\varphi(B_i)\varphi(B_{i+1})=\widehat{B}_i\widehat{B}_{i+1}$ is an edge in $\aux(Y,t)$. 
	Now, suppose that $B_i$ and $B_{i+1}$ are contained in distinct parts of $G_X$ and $G_Z$ and assume w.l.o.g. that $U_i$ is a block in $G_Y$.
	By Condition 4.a in the definition of $\equiv_t$, we have $V(B_i)\cap B_t = V(\widehat{B}_i)\cap B_t$.
	We deduce that $\widehat{B}_i$ and $B_{i+1}=\widehat{B}_{i+1}$ share a vertex in $B_t$.
	From the sequence $\widehat{B}_1,\dots,\widehat{B}_{m'}$, we conclude that there exists a path in $F$ between $U_1-c$ and $U_2-c$ in $F-c$.
	This contradicts the assumption that $c$ is a cut vertex. 
	
	\medskip
	
	(2) ($F$  contains a vertex of $S$ if and only if $B$ contains a vertex of $S$.)
	Observe that 
	for each block $U$ of $\aux(Y,t)$, $U$ contains a vertex of $S$ if and only if $\varphi(U)$ contains a vertex of $S$, and
	for every edge $uv$ of $\aux(Y,t)$, $M_{uv}$ has a vertex of $S$ if and only if $M_{\varphi(u) \varphi(v)}$ has a vertex of $S$. Since $F\cap G_Z=B\cap G_Z$, we obtain the result.
	\medskip
	
	(3) ($F$ is bipartite if and only if $B$ is bipartite.)
	Suppose that $B$ is bipartite. 
	We take the bipartition $(L, R)$ of $B$. As a connected bipartite graph has a unique bipartition, 
	this is unique up to changing $L$ and $R$.
	
	As $F\cap G_Z=B\cap G_Z$, this gives a bipartition $(L', R')$ of $F\cap G_Z$.
	Let $u, v\in V(F)\cap B_t$ be vertices that are contained in the same connected component of $F\cap G_Y$. 
	Assume $u, v$ are contained in the same part of $L'$ and $R'$.
	Since $u, v$ are also contained in the same connected component of $B\cap G_X$ and they are in the same part of $L$ and $R$, 
	all the paths from $u$ to $v$ in $B\cap G_X$ have even length.
	Note that the blocks containing edges of the path from $u$ to $v$ are all contained in $B$, and those blocks appear in 
	$\aux_p(X, t)$. As $B$ is bipartite, each of these blocks is bipartite. 
	
	Since $\aux(Y, t)$ is isomorphic to $\aux(X, t)$, 
	there is a corresponding sequence of blocks whose last blocks contain $u$ and $v$, respectively, 
	and all these blocks are bipartite.
	By the last condition of the equivalence relation $\equiv_t$, 
	there is an even path from $u$ to $v$ in $G_Y$.
	This shows that a bipartition of $F\cap G_Y$ is compatible with the bipartition $(L', R')$ of $F\cap G_Z$. 
	Thus, $F$ is bipartite.
\end{proof}

	We are now ready to describe our algorithm.
	For each node $t$ of $T$ and $I\subseteq B_t$, 
	let $\cP[t, I]$ be the set of all partial solutions $X$ of $G_t$
	where $X\cap B_t=I$.
	A reduced set $\cR[t, I]$ is a subset of $\cP[t, I]$ satisfying that 
	\begin{itemize}
		\item for every partial solution $X\in \cP[t, I]$, 
		there exists $X'\in \cR[t, I]$ where $X\equiv_t X'$ and $w(X')\ge  w(X)$, and 
		\item no two partial solutions in $\cR[t, I]$ are equivalent.
	\end{itemize} 
	We will recursively compute a reduced set $\cR[t, I]$ for every node $t$ of $T$ and $I\subseteq B_t$.
	Claim~\ref{claim:boundequiv} guarantees that $\abs{\bigcup_{I\subseteq B_t}\cR[t, I]}= 2^{\mathcal{O}(\tw\log \tw)}$. 
	
	We describe how to compute a reduced set $\cR[t, I]$ depending on the type of the node $t$, 
	and prove the correctness and the running time of each procedure.
	We fix a node $t$ and $I\subseteq B_t$.
	For each leaf node $t$ and $I=\emptyset$, we assign $\cR[t,I]:=\emptyset$. 
	For $\cA\subseteq 2^{V(G_t)}$, we define $\reduce_t(\cA)$ as the operation which removes the elements of $\cA$ that does not induce $S$-bipartite graph and then returns a set that contains, for each equivalence class $\cC$ of $\equiv_t$ over $\cA$, a partial solution of $\cC$ of maximum weight.

	\medskip\medskip
	\noindent\textbf{1) $t$ is an introduce node with child $t'$ and $B_t\setminus B_{t'}=\{v\}$:}
	\medskip
	
	If $v\notin I$, then it is easy to see that $\cR[t', I]$ is a reduced set of $\cP[t,I]=\cP[t',I]$.
	In this case, we take $\cR[t,I]=\cR[t',I]$.
	
	Assume now that $v\in I$.
	We set $\cR[t,I]=\reduce_t(\cA)$ with $\cA$ the set that contains $X\cup \{v\}$ for every $X\in \cR[t',I\setminus \{v\}]$.
	
	We claim that $\cR[t, I]$ is a reduced set of $\cP[t, I]$.
	Let $X\in \cP[t, I]$. 
	As $v\in I$, we have that $X\setminus \{v\}\in \cP[t', I\setminus \{v\}]$.
	Since $\cR[t', I\setminus \{v\}]$ is a reduced set of $\cP[t', I\setminus \{v\}]$, 
	there exists $X'\in \cP[t', I\setminus \{v\}]$ such that 
	$X'\equiv_{t'}  X\setminus \{v\}$ and $w(X')\ge w(X\setminus \{v\})$.
	By the construction of $\cR[t, I]$, 
	we added $\widehat{X}$ to $\cR[t, I]$ where $\widehat{X}\equiv_t X'\cup \{v\}$ and $w(\widehat{X})\ge w(X'\cup \{v\})$.
	It is not difficult to check that $\widehat{X}\equiv_t X$ by considering $G_t$ and the sum $(G_{t'}, B_t\setminus \{v\})\oplus (G[B_t], B_t\setminus \{v\})$ and applying Claim 21.
	

	\medskip\medskip
	\noindent\textbf{2) $t$ is a forget node with child $t'$ and $B_{t'}\setminus B_t=\{v\}$:}
	\medskip
	
	We set $\cR[t,I]=\reduce_t(\cR[t', I]\cup \cR[t', I\cup \{v\}])$.
	We easily deduce that $\cR[t, I]$ is a reduced set of $\cP[t, I]$ from the fact that, by definition, $\cP[t,I]=\cP[t',I]\cup \cP[t',I\cup \{v\}]$.
	
	\medskip\medskip
	\noindent\textbf{3) $t$ is a join node with two children $t_1$ and $t_2$:}
	\medskip
	
	We set $\cR[t,I]=\reduce_t(\cA)$ where $\cA$ is the set that contains $X_1\cup X_2$ for every $X_1\in \cR[t_1,I]$ and $X_2\in \cR[t_2,I]$.

	We claim that $\cR[t, I]$ is a reduced set of $\cP[t, I]$.
	Let $X\in \cP[t, I]$. 
	For each $i\in \{1, 2\}$, 
	let $X_i:=X\cap V(G_{t_i})$.
	It is not difficult to see that $X_i\in \cP[t_i, I]$, as it is a partial solution of $G_{t_i}$.
	
	Since $\cR[t_i, I]$ is a reduced set of $\cP[t_i, I]$, 
	there exists $X_i'\in \cR[t_i, I]$ such that $X_i'\equiv_t X_i$ and $w(X_i')\geq w(X_i)$.
	By construction, $X_1'\cup X_2'\in \cA$, and there exists $X'\in \cR[t, I]$ such that 
	$X'\equiv_t X_1'\cup X_2'$ and $w(X')\geq w(X_1'\cup X_2')$.
	It is not difficult to check that $X'\equiv_t X$ by considering $G_t$ and the sum $(G_{t_1}, B_t)\oplus (G_{t_2}, B_{t_2})$ and applying Claim 21 twice.
	As $w(X_1'\cup X_2')\geq w(X_1\cup X_2)=w(X)$, we conclude that $\cR[t,I]$ is a reduced set.

\medskip
	
	It remains to prove the correctness and the running time of our algorithm.
	We can assume w.l.o.g. that the bag $B_r$ associated with the root $r$ of $T$ is empty.
	By definition $\cP[t,\emptyset]$ contains an optimal solution.
	Since $\cR[t,\emptyset]$ is a reduced set, we conclude that $\cR[t,\emptyset]$ also contains an optimal solution.
	
	For the running time, we use the fact that we can compute $\reduce_t(\cA)$ in time $\abs{\cA}2^{\Oh(\tw\log\tw)}n^2$.
	This follows from the upper bound of Claim \ref{claim:boundequiv} on the number of equivalence classes of $\equiv_t$ and the fact that for every partial solutions $X,Y$ of $G_t$, we can decide whether $X\equiv_t Y$ in time $\Oh(n^2)$.
	
	Let $t$ be a node of $T$ that is not a leaf.
	Observe that we set $\cR[t,I]=\reduce_t(\cA)$ where $\cA$ is some sets that depends on the type of the node $t$.
	Let $k$ be maximum size of a reduced set $\cR[t',I']$ for $t'\in V(T)$ and $I'\subseteq B_{t'}$.
	By construction, if $t$ is an introduce node, then the size of $\cA$ is at most $k$.
	If $t$ is a forget node, then the size of $\cA$ is at most $2k$.
	And if $t$ is a join node, then the size of $\cA$ is at most $k^2$.
	From Claim \ref{claim:boundequiv}, we have $k\leq 2^{\Oh(\tw\log\tw)}$.
	We deduce that, for every node $t$ and $I\subseteq B_t$, we can compute $\cR[t,I]$ in time $2^{\Oh(\tw\log\tw)}\cdot n^2$.
 	Since $T$ has at most $\Oh(n \cdot \tw^2)$ nodes and there are at most $2^{\Oh(\tw)}$ possibilities for $I$, the running time of our algorithm is $2^{\Oh(\tw\log\tw)}n^3$.
\end{proof}

	The dependency $n^3$ on the input size $n$ in Theorem~\ref{thm:algosoct}
	was obtained because we keep the partial solutions themselves, and have to check their equivalences.
	We believe that with a careful argument by keeping only auxiliary graphs $\aux(X, t)$, we can reduce the dependence of the input size;
	however, for simplicity, we present the running time with $n^3$ factor.
	
	Now, we solve the other problems.

\begin{theorem}
	  \sfvs, \sfes, \vmwc, and their weighted variants can be solved in time $2^{\Oh(\tw\log\tw)}n^3$ on $n$-vertex graphs with treewidth $\tw$.
\end{theorem}
\begin{proof}
	It is easy to check that a graph has no $S$-traversing cycle if and only if every block of size at least $3$ has no vertex of $S$.
	So, for \ssfvs, we can simply adapt the algorithm for \ssoct, to ignore checks for bipartiteness.
	A partial solution at a node $t$ is a subset $Y$ of $V(G_t)$ such that $G_t[Y]$ is a graph having no $S$-traversing cycles.
	
	We use the same auxiliary graph $\aux(Y, t)$, and use the equivalence relation obtained by removing the conditions for bipartiteness in the equivalence relation for \ssoct.
	More specifically, we say that two partial solutions are equivalent for \ssfvs
	if $X\cap B_t=Y\cap B_t$, and there is an isomorphism $\varphi$ from $\aux(X,t)$ to $\aux(Y,t)$ such that the following conditions are satisfied:
	\begin{itemize}	
		\item For every vertex $v$ in $\aux(X,t)$, $v$ is active if and only if $\varphi(v)$ is active.
		\item For every vertex $v$ in $\aux(X,t)$, $v$ is a block if and only if $\varphi(v)$ is a block.
		\item For every active cut vertex $v$ in $\aux(X,t)$, we have $\varphi(v)=v$.
		\item For every active block $B$ in $\aux(X,t)$:
		\begin{itemize}
			\item  $V(B)\cap B_t = V(\varphi(B))\cap B_t$, and
			\item $V(B)\cap S \neq \emptyset$ if and only if $V(\varphi(B))\cap S\neq \emptyset$.
		\end{itemize}
		\item For every edge $uv$ in $\aux(X,t)$:
		\begin{itemize}
			\item $V(M_{uv})\cap S\neq \emptyset$ if and only if $V(M_{\varphi(u)\varphi(v)})\cap S\neq \emptyset$.
		\end{itemize} 
	\end{itemize}
	It is straightforward to adapt an algorithm for \ssoct to \ssfvs with this equivalence relation.
	We conclude that \ssfvs admits a $2^{\Oh(\tw\log\tw)}n^3$-time algorithm.
	
	For the weighted variant of \vmwc, we use the reduction to \ssfvs described in the introduction.
	Given an instance $(G,S,k)$ of \textsc{Weighted} \vmwc with weight function $w : V(G)\to \mathbb{R}$, we construct an equivalent instance $(G',S',k)$ of \textsc{Weighted} \ssfvs with $\tw(G')\leq \tw(G) + 1$ as follows.
	We add a new vertex $v'$ to $G$ adjacent to all the vertices in $S$.
	Let $G'$ be the resulting graph and $S'=\{v'\}$.
	Since we only add a single vertex, $\tw(G')\leq \tw(G)+1$.
	For every vertex $v$ in $G$, we set the weight of $v$ to be ``infinite'' if $v\in S\cup \{v'\}$ and $w(v)$ otherwise.
	Consequently, a solution cannot contain vertices in $S\cup \{v'\}$.
	One can easily prove that $(G,S,k)$ is a yes-instance of \textsc{Weighted} \vmwc if and only if $(G',S',k)$ is a yes-instance of \textsc{Weighted} \ssfvs.
	Hence, \vmwc and its weighted variant admit $2^{\Oh(\tw \log \tw)}n^3$-time algorithms.
	
	It remains to prove that the theorem holds for \textsc{Weighted} \sfes (\textsc{WSFES} for short).
	For doing so, we use a reduction to \textsc{Weighted} \ssfvs.
	Let $(G,S,k)$ be a instance of \textsc{WSFES} with weight function $w : E(G)\to \mathbb{R}$.
	We will construct a instance $(G',S',k)$ of \textsc{Weighted} \ssfvs where $\tw(G')= \tw(G)$.
	
	Let $G'$ be the graph obtained by subdividing the edges of $G$.
	Since subdivisions do not increase the treewidth, we have $\tw(G')=\tw(G)$.
	For every edge $e$ of $G$, we call $v_e$ the vertex in $G'$ created from the subdivision of $e$.
	Let $S'$ be the set $\{v_e \mid e \in S\}$.
	We give ``infinite'' weights to the vertices of $G'$ that belong to $V(G)\cup S'$.
	Moreover, for every edge $e\in E(G)\setminus S$, we give to $v_e$ the same weight as $e$.
	Consequently, a solution cannot contain vertices in $V(G)\cup S'$.
	It is easy to prove that $(G,S,k)$ is a yes-instance of \textsc{WSFES} if and only $(G',S',k)$ is a yes-instance of \textsc{Weighted} \ssfvs.
	We conclude that \textsc{WSFES} admits a $2^{\Oh(\tw \log \tw)}n^3$-time algorithm.
\end{proof}



\newpage

\appendix

\section{Semi-regular instances of \textsc{$k \times k$-Clique}}


Here we show the claimed ETH lower bound for \textsc{Semi-Regular $k \times k$-Clique}, namely that it is not easier than \textsc{$k \times k$-Clique}.
We insist that we do not need this result in the paper, but we believe that it may be helpful in a different context.
We recall that by \emph{semi-regular}, we mean that every vertex of a given column has the same degree towards another fixed column.

  The known parameterized reduction from \textsc{$k$-Clique} on general graphs to \textsc{$k$-Clique} on regular graphs will not work here, since it would increase the parameter polynomially (which we cannot afford).
  We take a couple of steps back and show how to partition the vertex set $V(G)$ of a hard instance of \textsc{Bounded-Degree $3$-Coloring} in such a way that the seminal reduction from \textsc{$3$-Coloring} to \textsc{$k \times k$-Clique} only produces semi-regular instances.
  Let us recall that the latter reduction builds one vertex per $3$-coloring of a part of the partition, and links by an edge every pair of consistent partial colorings (i.e., the union of the colorings is proper in the graph induced by the two parts).
  If the number of parts $k$ is chosen so that $3^{|V(G)|/k} \approx k$, we obtain a ``square'' $k$-by-$k$ instance of \textsc{Clique}.
  Now we want to ensure, in addition, that each $3$-coloring of a part $P$ has the same number of consistent $3$-colorings of another part $P'$.
  
  It is a folklore consequence of the Sparsification Lemma~\cite{sparsification} and known reductions that 3-coloring a~bounded-degree $n$-vertex graph cannot be done in $2^{o(n)}$.
  For instance Cygan et al. show the following.
  \begin{theorem}[Lemma 1 in \cite{Cygan17}]\label{hardness:3col}
    Unless the ETH fails, \textsc{3-Coloring} on $n$-vertex graphs of maximum degree 4 cannot be solved in $2^{o(n)}$.
  \end{theorem}
  Ultimately this result is a reduction from \textsc{$3$-SAT}.
  The \textsc{3-Coloring} instances produced by that reduction serve as our starting point. 

  \begin{theorem}
    Unless the ETH fails, \textsc{Semi-Regular $k \times k$-Clique} cannot be solved in time $2^{o(k \log k)}$.
  \end{theorem}
  \begin{proof}
    Let $G$ be a hard instance of \textsc{3-Coloring} with maximum degree 4, produced by the reduction of \cref{hardness:3col}.
    $G^2$, the square of $G$, (with an edge between two vertices at distance at most~2) has degree 16.
    By Hajnal-Szemerédi Theorem, $G^2$ admits an equitable coloring using 17 colors.
    This equitable coloring can further be found in polynomial time, by Kierstead and Kostochka~\cite{Kierstead08}.
    We refine this 17 classes arbitrarily into $k$ classes of equal size, such that $\lceil k \log_3 k \rceil = |V(G)|$.
    Let us call $\mathcal P$ the obtained partition of $V(G)$.
    We perform the reduction of \textsc{$3$-Coloring} to \textsc{$k \times k$-Clique} with this equipartition $\mathcal P$.
    The resulting instance is semi-regular.
    By design, for every pair of parts $P \neq P' \in \mathcal P$, $G[P \cup P']$ consists of some isolated edges between $P$ and $P'$, and isolated vertices.
    Indeed an edge within $P$, or within $P'$, or a degree-2 vertex would all contradict the coloring of $G^2$ (with 17 colors). 
    Thus any $3$-coloring of $P$ is consistent with the same number of $3$-colorings of $P'$.
    This shared number is 3 times the number of isolated vertices of $G[P \cup P']$  in $P'$ times twice the number of edges in $G[P \cup P']$.
  \end{proof}
  

\end{document}